\documentclass[journal,twoside]{IEEEtran}
\usepackage{amsmath,amsthm,amsfonts,amssymb,latexsym,extarrows,booktabs,array,arydshln}
\usepackage{fancyhdr,graphicx,tabularx}
\usepackage{multirow}
\usepackage{makecell}
\usepackage{amsfonts}
\usepackage{graphicx}
\usepackage{subfigure}
\usepackage {cite}
\usepackage {url}
\usepackage {stfloats}
\usepackage {mathrsfs}
\usepackage{algorithm}
\usepackage{algorithmic}
\usepackage{xcolor}
\usepackage{longtable}
\usepackage[OT2,T1]{fontenc}

\IEEEoverridecommandlockouts
\newtheorem{definition}{Definition}
\newtheorem{proposition}{Proposition}

\newtheorem{assumption}{Assumption}
\newtheorem{problem}{Problem}
\newtheorem{theorem}{Theorem}
\newtheorem{remark}{Remark}
\newtheorem{fact}{Fact}

\newtheorem{procedure}{Procedure}
\begin{document}
 \date{}
   \title{Hierarchical 2-D Feature Coding for Secure Pilot Authentication in  Multi-User Multi-Antenna \\OFDM  Systems: A  Reliability Bound \\ Contraction Perspective}
  \author{Dongyang~Xu,~\IEEEmembership{Student Member,~IEEE,}
        Pinyi~Ren,~\IEEEmembership{Member,~IEEE,}
        and James~A.~Ritcey,~\IEEEmembership{Fellow,~IEEE}
}
\maketitle
\begin{abstract}
Due to the \emph{publicly-known} and \emph{deterministic} characteristic of pilot tones,   pilot authentication (PA)  in multi-user multi-antenna OFDM systems is very suspectable to the jamming/nulling/spoofing behaviors.  To solve this,  we in this paper develop a hierarchical 2-D feature (H2DF) coding theory that  exploits the hidden pilot signal features,  i.e., the energy feature and independence feature, to secure pilot information coding  which is  applied  between legitimate  parties   through a well-designed five-layer  hierarchical coding (HC) model  to  achieve secure multiuser PA (SMPA).    The reliability of SMPA  is  characterised using  the identification error probability (IEP) of pilot encoding and decoding, with  the exact closed-form  upper and lower bounds. However, this phenomenon of  non-tight bounds brings about the risk of  long-term  instability  in SMPA. Therefore,  a reliability bound contraction (RBC) theory  is developed  to shrink the bound interval and  practically, this is done by an easy-to-implement technique, namely,  codebook partition  within the H2DF code. In this process,  a  tradeoff  between the upper  and lower  bounds of IEP is identified and  a problem of  optimal upper-lower bound  tradeoff  is  formulated, with the objective of   optimizing the cardinality of sub-codebooks such that the upper and lower bounds coincide.   Solving this,  we finally derive an exact closed-form  expression for IEP, which realizes  a stable and highly-reliable SMPA. Numerical results validate the stability and resilience of H2DF coding  in  SMPA.
\end{abstract}
\IEEEpeerreviewmaketitle
\section{Introduction}
\label{introduction}
\IEEEPARstart{R}{adio} security, either from a tactical perspective or in a commercial viewpoint, has drawn increasing attentions in wireless communication systems. The sophisticated characteristic of radio channels, such as the open and shared nature, create an operating environment vulnerable to intentional information security attacks that target specific radio technologies~\cite{Bogale}. Orthogonal frequency-division multiplexing (OFDM) technique becomes such a  typical victim when  it plays an increasing role in modern wireless systems,  standards (e.g., LTE, 802.11a/n/ac/ax/ah) or even under  tactical scenarios~\cite{Shahriar}. Without  comprehensive  precautions against attacks, OFDM technique  comes to be sensitive and  fragile in the respect of its waveform transmission and receiving which is very vulnerable to various  physical-layer attacks~\cite{Rahbari,Lichtman}. This paper investigates  the pilot-aware attack  on the channel estimation  process in multi-antenna OFDM communications~\cite{Xu_Optimal}. Conventionally, channel estimation  is performed with high accuracy by  using the publicly-known and deterministic  pilot tones  that are shared on the time-frequency resource grid (TFRG) by  all parties~\cite{Ozdemir}. Basically, the estimation performance  is   guaranteed  by perfect pilot authentication (PA)~\cite{Xu_CF}, since the authentication signal~\cite{Tu,Lai}, i.e., a unique pilot  tone from one certain legitimate user (LU),  is verified, therefore, known at the receiver (named Alice), and finally is enabled for precise  channel estimation that belongs to the LU. In other words, guaranteeing an exact and unique pilot tone for one LU means authenticating the authenticity of its channel state information (CSI), if estimated. However, PA mechanism lacks  specialized protections from the beginning and   a  pilot-aware attacker, named Ava,  can  easily jamm/null/spoof  those publicly-known  pilot tones by  launching pilot tone jamming (PTJ) attack~\cite{Clancy1,Clancy2}, pilot tone nulling (PTN) attack~\cite{Sodagari} and pilot tone spoofing (PTS) attack~\cite{Xu_ICA}. Finally  the channel estimation process at  Alice is seriously paralyzed.

\subsection{Related Works}
Basically, secure PA  here refers to  confirming  the authenticities of pilot tones from LUs  suffering above three  attacks.  This  includes how to detect any alteration to their authenticities  and how to protect and further maintain high authenticities.  Since  PA also means  authenticating CSIs, much work have been extensively investigated  on this area,  from narrow-band single-carrier systems~\cite{Zhou,Kapetanovic1,Wu1,Tugnait2,Xiong1,Kapetanovic2,Kang,Adhikary}  to wide-band multi-carrier systems~\cite{Xu_Optimal,Xu_CF,Clancy1,Clancy2,Sodagari, Xu_ICA,Shahriar2}.

The issue in PA in  narrow-band single-carrier systems was introduced  in~\cite{Zhou}  in which  a pilot contamination (PC) attack, one type  of  PTS attack,  was  evaluated. Following~\cite{Zhou},  much work were studied,  but  limited to  detecting the  alteration to pilot authenticities  by exploiting the physical layer information, such as auxiliary training or data sequences~\cite{Kapetanovic1,Wu1,Tugnait2,Xiong1} and some prior known channel information~\cite{Kapetanovic2,Kang, Adhikary}.  The  issue in PA in multi-subcarrier scenarios was first presented  by Clancy et al.~\cite{Clancy1}, verifying  the  possibility and effectiveness of PTJ attack. Following this,  PTJ attack was then studied for  single-input single-output (SISO)-OFDM communications in~\cite{Clancy2} which also introduced the PTN attack and then extended it to the  multiple-input multiple-output (MIMO)-OFDM system~\cite{Sodagari}.  The initial  attempt  to  safeguard  PA  under pilot aware attack  was proposed  in~\cite{Shahriar2}, that is, transforming the PTN and PTS attack into   PTJ attack by  randomizing the  locations and values of regular pilot tones on TFRG.   It  figured out the importance of random  pilot tone scheduling for avoiding the pilot aware attack.  Hinted by this,  authors in~\cite{Xu_ICA}, for a single-user scenario, proposed a coding based PA framework under PTS attack by exploiting pilot randomization  and a subcarrier-block discriminating coding (SBDC)  mechanism. In~\cite{Xu_Optimal}, the authors considered a  practical one-ring scattering  scenario in which a specific  spatial fading  correlation model, rather than a general form in~\cite{Xu_ICA}, was investigated.  They also proposed an independence-checking coding (ICC) theory for which SBDC could be just seen as its special form.  However,  the SBDC and ICC method can only differentiate two nodes (including Ava) at most since one more node will incur confusion on the discriminative feature, basically a binary result (e.g., the number of 1 digit is more than that of 0 digit, or not.), mentioned therein.  Out of  consideration for this,  authors in~\cite{Xu_CF} considered a two-user scenario and proposed a code-frequency block group (CFBG) code to support PA between two LUs. It introduced the necessity of a three-step solution, including pilot conveying, separation and identification.  The biggest problem is that when randomly-imitating attack happens, the code is invalidated and PA  then highly relies on  the difference between  spatial fading correlations of LUs and Ava whose correlation  model is generally hard to acquire.  If Ava has the same  correlation property  with one LU, for example, it has  the  same mean  angle of arrival (AoA) as one certain LU, the PA for that LU is also paralysed completely.
 \begin{figure*}[!t]
\centering \includegraphics[width=1\linewidth]{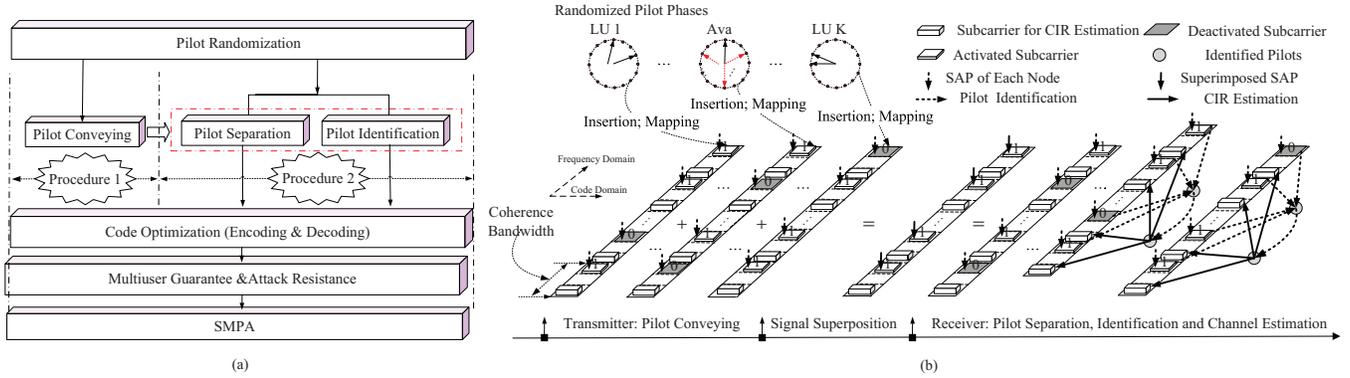}
\caption{(a) Diagram of the general procedures for  SMPA; (b) Specific procedures for SMPA. }\label{General_Architecture}
\end{figure*}
\subsection{Motivations}
The above observations  prompt us to establish a  secure  multiuser PA (SMPA)  mechanism  from the point of view of a pure coding approach and  also a  multi-user perspective.  As shown in Fig.~\ref{General_Architecture}~(a), pilot randomization is a prerequisite. Then the procedures of pilot conveying, separation and identification~\cite{Xu_Optimal, Xu_CF} are adopted, but with extra basic  considerations.
\begin{procedure}[Pilot Conveying]
Selectively activating and deactivating the OFDM subcarriers to create various subcarrier activation pattern (SAP) candidates; Encoding SAPs in such a way that those SAPs can carry pilot information in the form of codewords;
\end{procedure}
We emphasise that the pilot information in this paper refers to the pilot phases which are then randomized.  More clearly and intuitively,   the overall  process is depicted in Fig.~\ref{General_Architecture}~(b) and described as follows:

We  insert multiuser pilot tones into subcarriers on TFRG  in such a way  that  every single pilot subcarrier for SAP and  those for frequency-domain subcarrier (FS) channel estimation ( thus for channel impulse response(CIR) estimation) are located within
the range of coherence bandwidth but at different frequency-domain positions.  For the sake of simplicity, we configure  one pilot subcarrier for FS channel  estimation and one  paired pilot subcarrier for  SAPs. This operation guarantees the mutual independence of FS channels among adjacent  positions of each SAP.

On this basis, each LU  independently conveys their own pilot phase  in the form of encoded SAPs which are programmed  by codewords.  The specific principle is that   if  the $j$-th digit of the codeword  is equal to 1, the pilot tone signal  is inserted on the $j$-th subcarrier,  otherwise  this subcarrier will be  idle. In what follows,  pilot  separation and   identification naturally means codeword  separation and   identification.

In this context, the attacks will be transformed from PTJ, PTS, PTN  into the following hybrid  mode:
\begin{problem}[\textbf{Attack Model}]
A hybrid attack will include:
\begin{enumerate}
\item \textbf{Silence Cheating (SC):} Ava  keeps silence to misguide Alice since Alice cannot  recognize the non-existence of attacks.
 \item \textbf{Wide-Band Pilot (WB-PJ):} Ava  activates  the whole available subcarriers and thus launches WB-PJ attack to interfere LUs.  Therefore, the interpreted codeword at Alice is  a vector with all elements``1'',  which  carries no information.
\item \textbf{Partial-Band Pilot Jamming (PB-PJ):} Ava arbitrarily activates part of the subcarriers and launches  PB-PJ attack. The codeword interpreted  from  the observation subcarriers at A Jamming lice is seriously interfered and  misguided  if no special coding measure is taken.
 \item \textbf{Unpredictability:}  Ava could learn  the pilot tones employed by  each of  LUs  in advance and   jamming/spoofing/nulling the  pilot tones  of arbitrary one LU of interest. This is done by  searching  the list of  target LUs in store for attacking. This list is only known by Ava and unpredictable for both Alice and LUs.
\end{enumerate}
\end{problem}
Now the security goals  require  not only maintaining PA among LUs but also protecting those  established PA  from being attacked. We can see that PA is a probabilistic event and the security goal turns to be  the reliability of  pilot encoding/decoding.

We denote the  first requirement  by the  \textbf{Multiuser Guarantee} which is demonstrated in Problem 2 and denote  the second one by the \textbf{Attack Resistance} for which the attack model is given in Problem 1.  The relationship among pilot randomization,  multiuser guarantee and attack resistance is depicted in Fig.~\ref{General_Architecture}~(a).

\begin{problem}[\textbf{Multiuser Guarantee}]
The multiuser guarantee that is ensured by codewords includes three aspects:
\begin{description}
\item[\textbf{P. 2.1}] \textbf{Unique Pre-Separation Identification  (UPrSI):} To guarantee that each codeword has a unique identifer.
 \item[\textbf{P. 2.2}] \textbf{Uniquely Decipherable (UD):}
 \textbf{P. 2.2.1:} To  guarantee that each  superposition of up to  $K$ different codewords is  unique.
\textbf{P. 2.2.2:} To guarantee that each of the superimposed codeword can be  correctly decomposed   into a unique set of $K$ codewords.
\item[\textbf{P. 2.3}] \textbf{Unique Post-Separation Identification (UPoSI):}To ensure that  each of decomposed codewords is identified uniquely.
\end{description}
\end{problem}

 For  the second procedure to be designed, we stress that multiuser guarantee and attack resistance must be considered.
\begin{procedure}[Pilot Separation and Identification: A Mathematical Problem]
Those codewords for pilot conveying should be optimized such that  those codewords, though overlapped with each other (\textbf{Multiuser Guarantee})  and/or even  disturbed by Ava (\textbf{Attack Resistance}),  can be separated and identified with high reliability, thus decoded into the original pilots.
\end{procedure}
Having understood  above  procedures theoretically, we now turn to generally introduce the  practical procedures  as the Fig.~\ref{General_Architecture}~(b) indicates.   LUs and Ava  create  SAPs representing  their own randomized  pilot phases to be transmitted.  Those SAPs, after undergoing wireless channels, suffer from  the superposition interference from each other, and  finally  are superimposed and observed at Alice which separates and identifies  those pilots. This is  a basic process of multiuser PA.  Finally  those  authenticated pilots are utilized for channel estimation  using the estimator in~\cite{Xu_Optimal}. Until now, we have clarified  the procedures and key issues  for achieving SMPA.
\subsection{Contributions}
Solving  above issues requires  a reliable coding support. In a physical sense, the signals  from each node  carry  a lot of  features, such as, energy, independence and so forth, depending on how each node uses it. Different from the previous extra information, like spatial correlation information,   these signal features, when generated,  have  already been  hidden in the signals  and thus there is no need to provide them priorly  by system operators. The key is whether or not  we could  dig them out and how we use them.

For the first time, we propose  exploiting   those   signal  features to secure  information coding and  aim to answer the question, namely, \emph{can the hidden signal features  improve  the performance reliability  of  conventional coding  technique in attack environment?} We show the answer is yes,  and stress  that this novel and general comprehension towards coding technique constitutes the  core of our H2DF coding theory.

Before detailing our contributions, we need to clearly understand  what  type of signal structures Alice can employ, and recognize  the steps involved. In this paper, four basic steps are modeled, including  \textbf{1. extracting features, i.e., energy feature and independence feature}; \textbf{2. representing  features}; \textbf{3. encoding  features}; \textbf{4. decoding features}. Of all the  four steps, feature encoding and decoding  are  the core components determining  the final performance of SMPA mechanism.
Along  the lines of \textbf{Procedure 1} and \textbf{2}, we summarize the main contributions of this paper  as follows:
\begin{enumerate}
\item Basically and inevitably, we consider  examining the superposition  characteristics of multiple potential signals on each single pilot subcarrier.  We find that thanks to the  indelible and unique nature of the signal energy from each node, we could  \textbf{extract} and  \textbf{represent} the \textbf{energy feature} through the well-known energy detection technique  as the number of signals detected on each subcarrier.  We encode the derived number as  code digits (including binary digits) and construct a code-frequency domain on which  we  formulate feature encoding matrices in the form of  codebooks by deliberately grouping the digits into codewords.  Each binary codeword within codebooks  could precisely indicates how each of SAPs is triggered, thus achieving \textbf{Procedure 1}.
 \item\textbf{1)} We further  identify the second feature as the  \textbf{independence feature}  of  pilot signals from each  node. A differential coding technique is well designed  to fully  \textbf{extract} and  \textbf{represent}  this kind of  feature as the binary code. In this way,  the previous  feature encoding matrices is enabled to include the code information  of  both energy and  independence features.  The feature encoding matrices are optimized by  subtly coupling  the differential code with the cover-free code with the aim of supporting multiuser guarantee, which  constitutes the  \textbf{encoding} functionality  of  H2DF coding theory. \textbf{2)} For the  \textbf{decoding} functionality,  we construct a hierarchical decoding (HD) model   to  achieve attack resistance on the basis of  multiuser guarantee, which finally realizing \textbf{Procedure 2}.
\item  The  reliability  of  the overall encoding and decoding represents the resilience performance of  SMPA against attacks. To characterize  this metric,   we formulate the concept of identification error probability (IEP), bounded by the exact upper and lower  bounds.   This phenomenon of  bound fluctuation due to  the random selection of the codewords by each node indicates the long-term instability in  SMPA.  In order to reduce this instability,  a tradeoff between the upper and lower bounds is discovered, which prompts us to formally develop the bound contraction theory to further shrink  the bound interval.  A technique of  codebook partition is proposed to achieve this successfully and  an optimal upper-lower bound tradeoff  is realized.  Under this tradeoff,  an exact closed-form expression  of IEP is derived,  thus creating a  stable and highly-reliable SMPA performance.
\end{enumerate}

\emph{Organization:} In Section~\ref{MPPAA}, we present an overview of  pilot-aware attack on  multi-user PA in multi-antenna OFDM  systems. In Section~\ref{H2DFCE}, we introduce the encoding principle  of H2DF coding theory.   The decoding principle of H2DF coding theory  is described  in  Section~\ref{H2DFCD}.  A reliability bound  contraction theory   is  provided in Section~\ref{RBCT}.   Simulation results are presented in Section~\ref{SR} and finally  we conclude our work in Section~\ref{Conclusions}.

\emph{Notations:} We use boldface capital letters ${\bf{A}}$ for matrices, boldface small letters ${\bf{a}}$ for vectors , and small letters $a$ for scalars. ${{\bf{A}}^*}$, ${{\bf{A}}^{\rm{T}}}$, ${{\bf{A}}^{{H}}}$  and ${\bf{A}}\left( {:,1:x} \right)$ respectively denotes   the  conjugate operation, the transpose,  the conjugate transpose  and the first $x$ columns of matrix ${\bf{A}}$.  $\left\| {\cdot} \right\|$ denotes the Euclidean norm of a vector or a matrix. $\left| {\cdot} \right|$ is the cardinality of a set.  ${\mathbb{E}}\left\{  \cdot  \right\}$ is the expectation operator. $\otimes$ denotes  the Kronecker   product operator.   ${\rm{Diag}}\left\{ {\bf{a}} \right\}$   stands for the diagonal matrix with  the elements of column vector $\bf{a}$ on its diagonal.

\section{ Multi-User PA under Pilot Aware Attack:  Issues and Challenges}
\label{MPPAA}
\begin{figure}[!t]
\centering \vspace{-10pt}\includegraphics[width=0.75\linewidth]{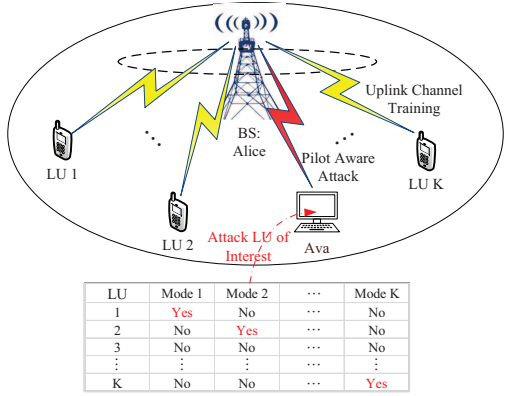}
\caption{System model of  $K$-user MISO-OFDM system under  the pilot aware attack in the uplink. }
\label{System_model}
\end{figure}
We in this section  outline a fundamental  overview of  multi-user PA issue under pilot aware attack, from a mathematical point of  view.  We will begin the overview by  introducing  the basic system  and  signal model, and then demonstrate this issue.   Besides this, we will  describe  the  advantage  of  pilot randomization in avoiding this issue and most importantly, identify  the  key challenge.
\begin{table*}\footnotesize
\begin{center}
\caption{Summary of Notations.}
\begin{tabular}{|l|l|}
\hline
Notations & Description \\
\hline
$N_{\rm T}$, $K$ & Number of transmitting antennas at BS and number of  LUs. \\
$ N^{\rm L}_{\rm E};N^{\rm A}_{\rm E}$  &    Number of subcarriers for CIR estimation:  for  the $m$-th LU;  for Ava.\\
$ N^{\rm L}_{\rm P}~(N^{\rm L}_{\rm P}=B);N^{\rm A}_{\rm P}$  &    Number of subcarriers for performing SAPs:   for  the $m$-th LU;  for Ava.\\
${\cal K}={\rm{\{ 1,}} \ldots {\rm{,}}K{\rm{\} }}$; ${\overline{\cal K}}$ & Index set of LUs; Index set of $K$ columns employed by $K$ LUs in $\bf B$.  \\
${\Psi^{\rm L}_{\rm E}}{\rm{ = }}\left\{ {{i_0},{i_1}, \ldots ,{i_{{N^{\rm L}_{\rm E}} - 1}}} \right\}$, ${\Psi^{\rm A}_{\rm E}}{\rm{ = }}\left\{ {{i_0},{i_1}, \ldots ,{i_{{N^{\rm A}_{\rm E}} - 1}}} \right\}$ & Index set of  subcarriers for CIR Estimation: for  the $m$-th LU;  for Ava. \\

${\Psi^{\rm L}_{\rm P}}{\rm{ = }}\left\{ {{i_0},{i_1}, \ldots ,{i_{{N^{\rm L}_{\rm P}} - 1}}} \right\}$, ${\Psi^{\rm A}_{\rm P}}{\rm{ = }}\left\{ {{i_0},{i_1}, \ldots ,{i_{{N^{\rm A}_{\rm P}} - 1}}} \right\}$ & Index set of  subcarriers  for performing SAPs and coding:  for  the $m$-th LU;  for Ava.\\

$x_{{\rm{L,}}m}^i\left[ k \right], {i \in {\Psi^{\rm L}_{\rm E}}}, {m \in {{\cal K}}}$;  $x_{\rm{A}}^i\left[ k \right], {i \in {\Psi^{\rm A}_{\rm E}}}$ &  Pilot tones at the $i$-th subcarrier and $k$-th symbol time: for the $m$-th LU; for Ava.  \\
$\rho _{{\rm{L}}, m}$, $\rho _{\rm{A}}$; ${\phi_{k,m}}$, $\varphi_{k,i}$ & Uplink  training power:  the $m$-th LU; for Ava; Pilot  phases: the $m$-th LU; for Ava. \\
${\bf{h}}_{\rm{L}, m}^i\in {{\mathbb C}^{L_{\rm s} \times 1}} $; ${\bf{h}}_{\rm{A}}^i \in {{\mathbb C}^{L_{\rm s}\times 1}}$  &  CIR  vectors, respectively from  the $m$-th LU and Ava to the $i$-th receive antenna of Alice. \\
$L_{\rm  s}$; $\sigma ^2$ &  Number of sampled multi-path taps in baseband, Average noise power of Alice.\\
${\bf {F}}\in {{\mathbb C}^{N_{\rm E}^{\rm L} \times N_{\rm E}^{\rm L}}}$; ${{\bf{F}}_{\rm{L}}}$; $T_c$ & DFT matrix; ${{\bf{F}}_{\rm{L}}} = \sqrt {N_{\rm E}^{\rm L}} {\bf{F}}\left( {:,1:L_{\rm s}} \right)$; Channel coherence time\\
 ${{\bf{v}}^i}\left[ k \right]\in {{\mathbb C}^{{N^{\rm L}_{\rm E}}\times 1}}$,  ${{\bf{v}}^i}\left[ k \right] \sim {\cal C N}\left( {0,{{{\bf{I}}_{N^{\rm L}_{\rm E}}}\sigma ^2}} \right)$  & Noise vector on time domain at the $i$-th antenna of Alice within the $k$-th symbol time. \\
 ${\cal A}$; ${\cal T} = \left\{ {{k_0}, \ldots ,{k_{{T_c} - 1}}} \right\}$ & $\left\{ {{\phi}:{{{\phi} = 2m\pi } \mathord{\left/
 {\vphantom {{{\phi _k} = 2m\pi } C}} \right.
 \kern-\nulldelimiterspace} C},0 \le m \le C - 1}, C=\left| {\cal A} \right| \right\}$; Set of OFDM symbols within $T_c$.  \\
 $M_{i}$ & Number of signals detected on the $i$-th subcarrier. \\
 ${{\bf{y}}_i}\left[ k \right]\in {{\mathbb C}^{{N_{\rm T}}\times 1}}$ &  Receiving signals  stacked at the $i$-th subcarrier  within the $k$-th OFDM symbol.\\
 ${{\bf{w}}_i}\in {{\mathbb C}^{{N_{\rm{T}}} \times 1}}$&  Noise signals  at the $i$-th subcarrier;  \\
 ${\bf{g}}_{k,i}^{{\rm{L}}}\in {{\mathbb C}^{{N_{\rm{T}}} \times 1}}$; ${\bf{g}}_i^{\rm{E}}\in {{\mathbb C}^{{N_{\rm{T}}} \times 1}}$ & Channel frequency response  vectors of the $k$-th LU and that of Ava at the $i$-th subcarrier. \\
   ${{n_i}}$; $\bf{c}$  &  Jamming  pilot symbols of Ava on the $i$-th subcarrier; Codeword of Ava. \\
  ${{\bf b}_{{\rm S},K}}$ and ${{\bf m}_{{\rm S},K}}$&  SP sum  and ASP sum of   H2DF codewords from  all LUs;  \\
   ${{\bf b}_{\rm I}}$ and ${\bf{m}}_{\rm I}$&  SP sum  and ASP sum of $\bf{c}$ with  H2DF codewords from  all LUs;  \\
${{{b}}_{{\rm{S}},K,i}}$; ${{{m}}_{{\rm{S}},K,i}}$; ${{{b}}_{{\rm{I}},i}}$;  ${{{m}}_{{\rm{I}},i}}$; $c_{i}$ & The $i$-th ($1\le i \le B$) element of ${{\bf{b}}_{{\rm{S}},K}}$, ${{\bf{m}}_{{\rm{S}},K}}$, ${{\bf{b}}_{\rm{I}}}$,  ${{\bf{m}}_{\rm{I}}}$  and $\bf{c}$. \\
 ${\overline{\cal{B}}}_{K}$ and ${\overline{\cal{M}}}_{K}$ & Set of all column vectors of ${\bf B}_{K}$ and ${\bf M}_{K}$;  \\
 $\cal D$ & Set of  position indices of digits in ${{\bf b}_{\rm I}}$. $\forall {i}\in {\cal D}$,  ${{{m}}_{{\rm{I}},{i}}}=1$.\\
\hline
\end{tabular}
\end{center}
\end{table*}
\subsection{System Description and Problem Model}
We consider a synchronous    multi-user  multiple-input single-output (MISO)-OFDM systems with a $N_{\rm T}$-antenna Alice and $K$ single-antenna LUs.  Here, pilot tone based multi-user channel estimation is considered in the uplink~\cite{Ozdemir}.  Conventionally,  multi-user PA is accomplished  by assigning LUs  with   \emph{publicly-known and deterministic} pilot tones that can be identified. This  mechanism is very  fragile and actually has no privacy. Without imitating the identities of LUs, Ava  merely with single antenna can  synchronously  interfere  pilot tones indexed  by ${\Psi^{\rm A}_{\rm E}}$ and launches the  behaviors shown in \textbf{Attack Model}.
\subsection{Signal Model}
In this subsection, we formulate  the receiving  signal model at Alice.  To begin with, we will give the concept of  pilot insertion pattern (PIP) which indicates  the way of  inserting   pilot tones  across  subcarriers and OFDM symbols.
\begin{assumption}[\textbf{Frequency-domain PIP}]
We in this paper assume $x^{j}_{{\rm{L,}}m}\left[ k \right] = {x_{{\rm{L,}}m}}\left[ k \right] = \sqrt {{\rho _{{\rm{L,}}m}}} {e^{j{\phi _{k,m}}}},\forall i, i \in {\Psi^{\rm{L}}_{\rm{E}}}, {m \in {{\cal K}}}$ for low overhead consideration and theoretical analysis. Alternatively,  we can superimpose $x^{i}_{{\rm{L,}}m}\left[ k \right] $ onto a dedicated  pilot sequence optimized under a non-security oriented scenario and utilize this new pilot for training. At this point,  $\phi_{k,m}$ can be an additional phase difference for security consideration. We do not impose the phase constraint on the PIP  strategies  of  Ava, that is,  $x^{i}_{{\rm{A}}}\left[ k \right] = \sqrt {{\rho _{\rm{A}}}} {e^{j{\varphi _{k,i}}}},i \in {\Psi^{\rm{A}}_{\rm{E}}}$.
\end{assumption}

Let us proceed to the basic OFDM procedure. First,  the  pilot tones of LUs and Ava  over $N^{\rm L}_{\rm E}$ subcarriers are  respectively  stacked as $N^{\rm L}_{\rm E}$ by $1$  vectors ${{\bf{x}}_{{\rm{L,}}m}}\left[ k \right] = \left[ {{x^{j}_{{\rm{L,}}m}}\left[ k \right]} \right]_{j \in \Psi^{\rm L}_{\rm E} }^{\rm{T}}$ and ${{\bf{x}}_{\rm{A}}}\left[ k \right] = \left[ {{x^{j}_{{\rm{A}}}}\left[ k \right]} \right]_{j \in \Psi^{\rm A}_{\rm E} }^{\rm{T}}$. Assume that the  length of cyclic prefix is larger than  the maximum length $L_{\rm s}$ of all channels.  The parallel streams, i.e.,  ${{\bf{x}}_{{\rm{L}},m}}\left[ k \right]$, $m \in \cal{K}$ and  ${{\bf{x}}_{{\rm{A}}}}\left[ k \right]$, are modulated with inverse fast Fourier transform (IFFT).

Then the time-domain $N^{\rm L}_{\rm E}$ by $1$  vector ${{\bf{y}}^i}\left[ k \right]$, derived by Alice after removing the cyclic prefix at the $i$-th receiving antenna, can be written as:
\begin{equation}\label{E.1}
{{\bf{y}}^i}\left[ k \right] = \sum\limits_{m = 1}^K {{\bf{H}}_{{\rm{C,}}m}^i{{\bf{F}}^{\rm{H}}}{{\bf{x}}_{{\rm{L,}}m}}\left[ k \right]}  + {\bf{H}}_{{\rm{C,A}}}^i{{\bf{F}}^{\rm{H}}}{{\bf{x}}_{\rm{A}}}\left[ k \right] + {{\bf{v}}^i}\left[ k \right]
\end{equation}
Here,  ${\bf{H}}_{{\rm{C}},m}^i$  is the $N^{\rm L}_{\rm E} \times N^{\rm L}_{\rm E}$ circulant matrices of the $m$-th LU, with the  first column  given by ${\left[ {\begin{array}{*{20}{c}}
{{\bf{h}}_{{\rm{L,}}m}^{{i^{\rm{T}}}}}&{{{\bf{0}}_{1 \times \left( {N^{\rm L}_{\rm E}- L_{\rm s}} \right)}}}
\end{array}} \right]^{\rm{T}}}$.
${\bf{H}}_{{\rm{C,A}}}^i$ is a $N^{\rm A}_{\rm E} \times N^{\rm A}_{\rm E}$ circulant matrix with the first column given by ${\left[ {\begin{array}{*{20}{c}}
{{\bf{h}}_{\rm{A}}^{{i^{\rm{T}}}}}&{{{\bf{0}}_{1 \times \left( {N^{\rm A}_{\rm E} - L_{\rm s}} \right)}}}
\end{array}} \right]^{\rm{T}}}$ and ${\bf{h}}_{\rm{A}}^i $ is  assumed to be independent with  ${\bf{h}}_{{\rm L},m}^i, \forall m \in \cal{K}$.

Taking  FFT,  Alice  finally derives the    frequency-domain $N^{\rm L}_{\rm E}$ by $1$ signal vector  at the $i$-th receive antenna as
\begin{equation}\label{E.3}
{\widetilde {{\bf{y}}}^i}\left[ k \right] =\sum\limits_{m = 1}^K {{{\bf{F}}_{\rm{L}}}{\bf{h}}_{{\rm{L}},m}^i{x_{{\rm{L}},m}}\left[ k \right]}  + {\rm{Diag}}\left\{ {{{\bf{x}}_{\rm{A}}}\left[ k \right]} \right\}{{\bf{F}}_{\rm{L}}}{\bf{h}}_{\rm{A}}^i + {{\bf{w}}^i}\left[ k \right]
\end{equation}
where  ${{\bf{w}}^i}\left[ k \right]={{\bf{F}}}{{\bf{v}}^i}\left[ k \right]$.
\subsection{Multi-User Channel Estimation Model}
 We only focus on the FS estimation model  under PTS attack mode. Ava  impersonates the $m$-th LU  by utilizing the same pilot tone learned.  In this case, there exists ${\Psi^{\rm{L}}_{\rm{E}}}\cup {\Psi^{\rm{A}}_{\rm{E}}} = {\Psi^{\rm{L}}_{\rm{E}}}$ and ${{{x}}^{i}_{{\rm{A}}}}\left[ k \right]={{{x}}_{{\rm{L}},m}}\left[ k \right] , \forall i, i \in {\Psi^{\rm{L}}_{\rm{E}}}$.
Stacking ${\widetilde {\bf{y}}^i}\left[ k \right]$ within $K$ OFDM symbol time, we can rewrite  signals in  Eq.~(\ref{E.3}) as:
\begin{equation}
{{\bf{Y}}}_{{\rm{PTS}}}^i = \sum\limits_{j= 1}^K {{{\bf{F}}_{\rm{L}}}{\bf{h}}_{{\rm{L}},j}^i{{\bf{x}}_{{\rm{L}},j}}}  + {{\bf{F}}_{\rm{L}}}{\bf{h}}_{\rm{A}}^i{{\bf{x}}_{{\rm{L}},m}} + {{\bf{W}}^i}
\end{equation}
where the $N^{\rm L}_{\rm E} \times K$   matrix ${\bf{Y}}_{{\rm{PTS}}}^i$ satisfies ${\bf{Y}}_{{\rm{PTS}}}^i = \left[ {\widetilde {\bf{y}}^i{{\left[ k \right]}_{1 \le k \le K}}} \right]$.  The $1 \times K$ vector ${{\bf{x}}_{{\rm{L}},m}}$ satisfies ${{\bf{x}}_{{\rm{L}},m}}{\rm{ = }}\left[ {{{x}_{{\rm{L}},m}}{{\left[ k \right]}_{1 \le k \le K}}} \right]$ and ${{\bf{W}}^i}$ is also a  $N^{\rm L}_{\rm E} \times K$  matrix with ${{\bf{W}}^i} = \left[ {{{\bf{w}}^i}{{\left[ k \right]}_{1 \le k \le K}}} \right]$.
For simplification, we exemplify the orthogonal pilots to demonstrate the influence of PTS  attack.   Given the orthogonal pilots with $ {{\bf{x}}_{{\rm{L}},m}}{\bf{x}}_{{\rm{L}},n}^ + = 0,\forall m \ne n$, a least square (LS) estimation of ${{{\bf{h}}^{i}_{{\rm L},m}}}$, contaminated by  ${{{\bf{h}}^{i}_{\rm{A}}}}$  with  a noise bias, can be given by:
\begin{equation}\label{E.7}
\widehat {\bf{h}} = \left\{ {\begin{array}{*{20}{c}}
{{{\bf{F}}_{\rm{L}}}{\bf{h}}_{{\rm{L}},1}^i + {{\bf{F}}_{\rm{L}}}{\bf{h}}_{\rm{A}}^i + {{\bf{W}}^i}{{\left( {{{\bf{x}}_{{\rm{L}},1}}} \right)}^ + }}&{if\, m=1}\\
{{{\bf{F}}_{\rm{L}}}{\bf{h}}_{{\rm{L}},2}^i + {{\bf{F}}_{\rm{L}}}{\bf{h}}_{\rm{A}}^i + {{\bf{W}}^i}{{\left( {{{\bf{x}}_{{\rm{L}},2}}} \right)}^ + }}&{if\, m=2}\\
 \vdots & \vdots \\
{{{\bf{F}}_{\rm{L}}}{\bf{h}}_{{\rm{L}},K}^i + {{\bf{F}}_{\rm{L}}}{\bf{h}}_{\rm{A}}^i + {{\bf{W}}^i}{{\left( {{{\bf{x}}_{{\rm{L}},K}}} \right)}^ + }}&{if\, m=K}
\end{array}} \right.
\end{equation}
where $\left( {\cdot} \right)^+ $ is  the Moore-Penrose pseudoinverse. As to describing  PTN  attack and PTJ attack, we can refer to   the mathematical interpretation  in~\cite{Xu_Optimal}. As we can see, the channel estimation is completely paralysed and unable to be predicted in advance. Note that this phenomenon  also occurs  even if non-orthogonal pilots are adopted. What's more, we stress that any  prior  pilot design clues given to resist attack would also give information away to Ava.
\subsection{Pilot Randomization  and  Key Challenge}
Pilot randomization can avoid the pilot aware attack without imposing  any prior information on the pilot design. The common method is to randomly  select  phase candidates. Note that each of the phase candidates is  defaultly mapped into a unique quantized  sample, chosen  from the set ${\cal A}$.
 Since phase information provides  the security guarantee, thus without the need of huge overheads,   we make the following assumptions:
 \begin{assumption}[\textbf{Time-domain PIP}]
 During  two adjacent OFDM symbol time, such as, $k_i,k_{i+1}$, $i\ge0$, two pilot phases ${\phi _{{k_i},m}}$ and ${\phi _{{k_{i+1},m}}}$  are kept with fixed phase difference, that is,   ${\phi _{{k_{i+1},m}}} - {\phi _{{k_{i},m}}} = \overline \phi_{m} $. Here, ${\phi _{{k_{i+1},m}}}$ and ${\phi _{{k_{i},m}}}$ are both random but $\overline \phi_{m}$  is deterministic and  publicly known.
 \end{assumption}
The value of $C$  affects  the reliability  of proposed SMPA architecture and  as discussed in the Procedure 2, pilot randomization  also brings the necessity of  novel coding   theory.


\section{Hierarchical 2-D Feature Coding for SMPA Architecture: Encoding Part}
\label{H2DFCE}
Basically, any coding strategy includes the encoding  and decoding part. In this section,  the  encoding part of H2DF coding  is formulated,  which embraces  three parts, that is, energy feature extraction,  energy feature representation (for satisfying Procedure 1) and the feature encoding (only provides  the multiuser guarantee of Procedure 2).
\subsection{Energy Feature Extraction}
\label{EFE}
 A commonsense is that  wireless signal energy  is indelible. Using  the technique of eigenvalue ratio based energy detection (ERED)  in~\cite{Kobeissi}, we hope to precisely measure  the number of aggregated signals at subcarriers. The number  represents the energy feature that we could extract and encode further.

 Let us focus on a physical phenomenon, that is, the signal (or energy) superposition on each single subcarrier. This will contribute to  our quantitative modelling  for the energy  feature. On one hand, if we examine the  SAPs  employed by each node, we could find   they  are random and mutually independent, leading to the occurrence and superposition of activated and deactivated subcarriers. In other words,  the number of signals on each single subcarrier and their identities  are completely  unknown and unpredictable. On the other hand,  this uncertainty extends to the case where  the attacker is involved and  configure  arbitrary SAPs to interfere LUs.  Therefore, each subcarrier may carry at most $K+1$ signals and at least no  signal, depending on the choices of $K+1$ SAPs.

To capture the variations  of the number of aggregated signals on arbitrary single subcarrier,  a $\left( {K + 2} \right) \times {N_{\rm{T}}}$ receiving signal matrix within $K + 2$ OFDM symbols, denoted by ${{\bf{Y}}_{\rm{D}}}$,   is created for energy  detection.  Given the normalized covariance matrix defined by $\widehat {\bf{R}} = \frac{1}{{{\sigma ^2}}}{\bf{Y}}_{\rm D}{{\bf{Y}_{\rm D}}^{\rm{H}}}$, we define  its ordered eigenvalues  by  ${\lambda _{{K+2}}} >  \ldots  > {\lambda _1} > 0$ and construct  the test statistics by $T = \frac{{{\lambda _{K+2}}}}{{{\lambda _{1}}}}\mathop {\gtrless}\limits_{{{{\overline{\cal H}}_0}}}^{{{{{ {\cal H}}_0}}}} \gamma$ where $\gamma$ denotes the  decision threshold.  The hypothesis ${{ {\cal H}}_0}$  means that there exist signals and  ${{\overline{\cal H}}_{0}}$ means the opposite.

Based on ${{\bf{Y}}_{\rm{D}}}$,  Eq. (49) in~\cite{Kobeissi}  provides  a  decision threshold  function $\gamma {\rm{ }} \buildrel \Delta \over = f\left( {{N_{\rm{T}}},K,{P_f}} \right)$, for measuring how many  antennas on one subcarrier  are required to achieve a certain probability of  false alarm denoted by  ${P_f}$.  Therefore, we could establish a single-subcarrier encoding (SSE) principle,  finally encoding the number of detected signals  into binary  and $M$-ary digits.
 \begin{definition}[\textbf{SSE}]
Given fixed $K$ and $N_{\rm T}$,  one subcarrier  can be precisely encoded if, for any $\varepsilon > 0$, there exists a positive number $\gamma \left( \varepsilon  \right)$  such that, for all  $\gamma  \ge \gamma \left( \varepsilon  \right)$, ${P_f}$  is smaller than $\varepsilon $.
\end{definition}
 We should note that $f\left( {{N_{\rm{T}}},K,{P_f}} \right)$ is a monotone decreasing  function of two independent variables, i.e., ${N_{\rm{T}}}$ and ${P_f}$ but a monotone increasing  function of  $K$. For a given probability constraint ${\varepsilon}^*$, we could always expect a lower bound $\gamma \left( {{\varepsilon ^*}} \right)$ satisfying $\gamma \left( {{\varepsilon ^*}} \right) = f\left( {{N_{\rm{T}}},K,{\varepsilon ^*}} \right)$. Under this equation, we could flexibly configure  ${N_{\rm{T}}}$, $K$ and $\gamma \left( {{\varepsilon ^*}} \right)$ to make ${\varepsilon}^*$  approach zero~\cite{Kobeissi}.   We also find that  the value of $\gamma$  achieving zero-${\varepsilon}^*$  is decreased with the increase of  ${N_{\rm{T}}}$, but increased with the increase of $K$.
   \begin{figure*}[!t]
\centering\includegraphics[width=1.0\linewidth]{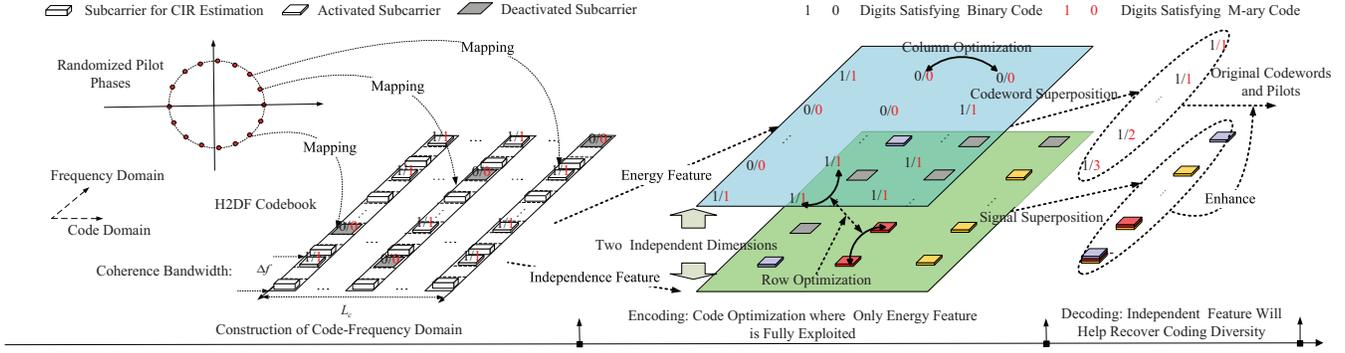}
\caption{ Construction of  code-frequency domain on which  a general description of H2DF encoding/decoding is depicted. }
\label{Pilot_encoding}
\end{figure*}

Using $\gamma \left( 0 \right)$ as the detection threshold, Alice  constructs new test statistics $T_{i} = \frac{{{\lambda _i}}}{{{\lambda _{1}}}}\mathop {\gtrless}\limits_{{{{\overline{\cal H}}_i}}}^{{{{{ {\cal H}}_i}}}} \gamma \left( 0 \right), 2\le i\le K+1$.  The hypothesis ${{ {\cal H}}_i}$  means $\left| {K +3- i} \right|$ signals coexist and  ${{\overline{\cal H}}_{i+1}}$ means the opposite. Using this, Alice  can determine the number of coexisting signals on arbitrary  single subcarrier. For example,  two  signals  are recognized only  when  both ${{ {\cal H}}_{K+1}}$  and ${{\overline{\cal H}}_{K}}$ hold true.

\subsection{Energy Feature Representation}
Basically, we could derive two types of code representing the energy feature. The first one is \textbf{binary  code} and the other one is \textbf{$M+1$-ary code}.
To formulate the code, we begin by constructing  the code-frequency domain.
\subsubsection{Code-Frequency Domain}
We denote the binary  digit corresponding to the $ m$-th pilot subcarrier   by ${s_{1, m}} $ satisfying:
 \begin{equation}\label{E.32}
{s_{1, m}} = \left\{ {\begin{array}{*{20}{c}}
1&{M_{m} \ge 1}\\
0&{otherwise}
\end{array}} \right.
  \end{equation}
Naturally, the  $\left( {M + 1} \right)$-ary digit is defined by:
  \begin{equation}
  {s_{2, m}} = M_{m}
   \end{equation}
Furthermore,  we denote a binary code vector set by ${{\cal S}_{1}}$  with ${{\cal S}_{1}}= \left\{ {\left. {\bf{s}_{1}} \right|{s_{1,m}} \in \left\{ {0,1} \right\},1 \le m \le {{L_c}}} \right\}$ where ${L_c}$ denotes the maximum length of  the  code.  Similarly, we   denote the $\left( {M + 1} \right)$-ary code vector set  by ${{\cal S}_{2}}$  satisfying   ${{\cal S}_{2}}= \left\{ {\left. {\bf{s}_2} \right|{s_{2,m}} \in \left\{ {0,\ldots, M} \right\},1 \le m \le {{L_c}}} \right\}$.

Finally, a code frequency domain with hybrid binary and $\left( {M + 1} \right)$-ary code digits can be formulated as a set of pairs $\left( {{{\bf{s}}},b} \right)$ with ${\bf{s}} \subset {{\cal S}_{1}} \cup {{\cal S}_{2}}$ and $1\le b \le B$ where $b$ is an integer  representing the subcarrier  index of appearance of the code. The construction process  can be depicted in Fig.~\ref{Pilot_encoding}.
\subsubsection{Achieving Procedure 1}
Grouping  the code digits  on code-frequency domain, we can derive  two types of  codes.
\begin{definition}
We call a $B\times C$ binary matrix satisfying  ${\bf{B}} = \left[ {{b_{j,i}}} \right]_{1 \le j \le B,1 \le i \le C}, {b_{j,i}} \in {\bf{s}} \subset {{\cal S}_{1}}$  and a  $B\times C$ $\left( {M + 1} \right)$-ary matrix  satisfying ${\bf{M}} = {\left[ {{m_{j,i}}} \right]_{1 \le j \le B,1 \le i \le C}}, {m_{j,i}} \in {\bf{s}} \subset {{\cal S}_{2}}$ as  the feature encoding matrices.
The $i$-th column of ${\bf B}$ and ${\bf M}$  are respectively denoted by  ${\bf b}_{i}$ and ${\bf m}_{i}$ with ${{\bf{b}}_i} = {\left[ {\begin{array}{*{20}{c}}
{{b_{1,i}}}& \cdots &{{b_{B,i}}}
\end{array}} \right]^{\rm{T}}}$ and ${{\bf{m}}_i} = {\left[ {\begin{array}{*{20}{c}}
{{m_{1,i}}}& \cdots &{{m_{B,i}}}
\end{array}} \right]^{\rm{T}}}$.   We call ${{\bf{b}}_i}$ a codeword  of $\bf B$ of  length $B$ and  ${{\bf{m}}_i}$ a codeword  of $\bf M$ with the same length..
\end{definition}
Each  LU could represent its energy feature  using binary codeword which also indicates its corresponding SAP.
\subsubsection{Codeword Arithmetic Principle}Two arithmetic operations between codewords are formulated, depending on  the specific code definition.
\begin{definition}
The superposition (SP) sum  ${{\bf{z}}_{i,j}} = {{\bf{b}}_{i}}{{\bigvee}}{{\bf{b}}_{j}}, 1\le i, j\le {C}$ (designated as the digit-by-digit Boolean sum) of two $B$-dimensional binary codewords  is defined by:
  \begin{equation}\label{E.33}
{z_{i,j,k}} = \left\{ {\begin{array}{*{20}{c}}
0&{if\,\,{b_{k,i}} = {b_{k,j}} = 0}\\
1&{otherwise}
\end{array}} \right., \forall  1\le k\le {B}
  \end{equation}
  where ${z_{i,j,k}}$ represents the $k$-th element of vector ${{\bf{z}}_{i,j}}$.   We say that a binary vector $\bf x$ \textbf{covers} a binary vector $\bf y$ if the Boolean sum satisfies  ${\bf{y}}{{\vee}}{\bf{x}}={\bf{x}} $
  \end{definition}

  \begin{definition}
 The algebraic superposition (ASP) sum (designated as the digit-by-digit  sum)  is defined by  ${\bf{d}}_{i,j} = {{\bf{b}}_{i}}{{+}}{{\bf{b}}_{j}}, 1\le i,j\le {C}$ in which two $B$-dimensional $\left( {M + 1} \right)$-ary codewords satisfy:
  \begin{equation}\label{E.34}
{{ d}_{i,j,k}} ={{ m}_{k,i}} + {{ m}_{k,j}}, \forall  1\le k\le {B}
  \end{equation}
   where  ${{ d}_{i,j,k}}$ denotes the  $k$-th element of vector ${\bf{d}}_{i,j}$.
 \end{definition}
\begin{figure*}[!t]
\centering \includegraphics[width=1.0\linewidth]{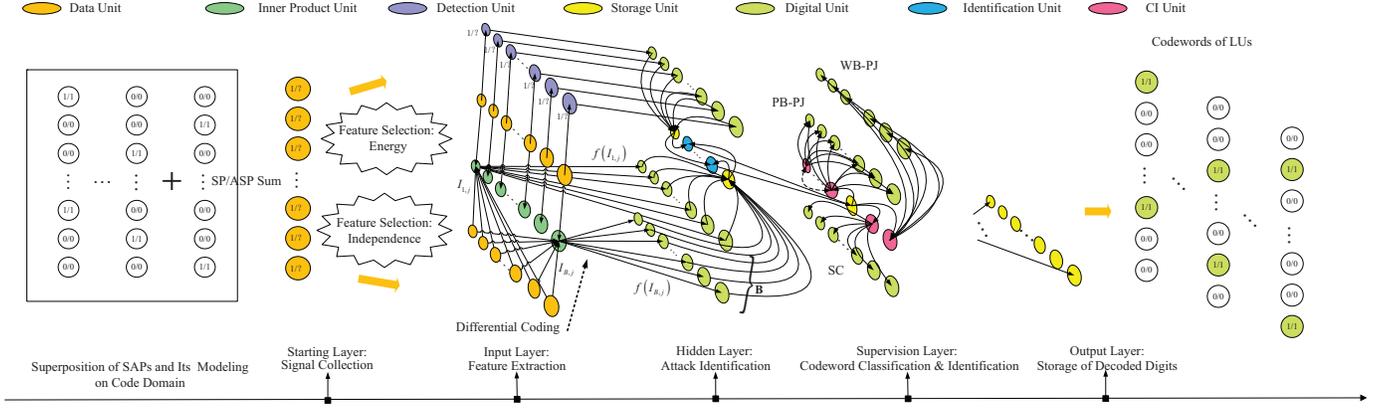}
\caption{HD  model. The starting layer performs signal collection which is interpreted as a superposition process of $K+1$ SAPs.  The same  signal information, divided into two paths, are then passed  to the data units in the  input  layer. For one path,  the energy feature extraction is guaranteed and for the other one,  the independence feature extraction is done. The resulted digits are propagated forward through hidden  and supervision layer,  finally facilitating the presentation of   $K$ codeword vectors of LUs in the output layer.}
\label{MLR}
\end{figure*}
\subsection{Feature Encoding: Coupling  Independence Features with Coding Diversity}
Here we  aim to further optimize the feature encoding  matrix.  As shown in Fig.~\ref{Pilot_encoding}, this is done, 1) by creating  the potential  diversity of cover-free coding   across the columns of binary matrix; 2) integrating the independence feature of receiving signals into its rows and coupling it with  cover-free coding  naturally. We  begin our  discussion by introducing the fundamental notion:
\begin{definition}
A $B\times C$ binary matrix ${\bf{B}}{\rm{ = }}{\left[ {{b_{j,i}}} \right]_{{\rm{1}} \le j \le B,{\rm{1}} \le i \le C}}$  is called a H2DF-$\left( {K,L,B} \right)$ code of length $B$, size $C$ and order $K$, if the following conditions are satisfied:

1) The arithmetic operation among codewords in $\bf B$  obeys  the SP principle.

2) \textbf{Column-Wise  Cover-free Coding (Column Optimization):} For arbitrary  two sets of columns, i.e., ${\cal P},{\cal Q} \subset \left\{ {1,2, \ldots ,C} \right\}$ such that  $\left| \cal P \right| = K$, $\left| \cal Q \right| = L$, and ${\cal P} \cap {\cal Q} = \emptyset $, there exists a row $i \in \left\{ {1,2, \ldots ,B} \right\}$ such that ${b_{i,j}} = 0,\forall j \in  \cal P $ and ${b_{i,j^{'}}} = 1,\forall j^{'} \in \cal Q $.

3) \textbf{Per-Word Independence-Aided Differential Coding (Row Optimization):} For any two  positions, i.e.,  $i$, $j$, on the frequency domain,  one within $T_{c}$ there exists:
\begin{equation}
\bigvee\limits_{l \in {{\overline {\cal K}}}} {\left( {{b_{i,l}} \oplus {b_{j,l}}} \oplus 1 \right)}  = {d_{i,j}}
\end{equation}
where ${d_{i,j}} = f\left( {{I_{i,j}}} \right),{I_{i,j}} = \left\langle {\frac{{{{\bf{y}}_i}\left[ k \right]}}{{\left\| {{{\bf{y}}_i}\left[ k \right]} \right\|}},\frac{{{{\bf{y}}_j}\left[ k \right]}}{{\left\| {{{\bf{y}}_j}\left[ k \right]} \right\|}}} \right\rangle, \forall k, k \in {\cal T}$.  $ {d_{i,j}}$ denotes the differential code and $\left\langle {\cdot} \right\rangle $ denotes  inner product operation.  $f$ represents  the differential encoder with decision  threshold  $\gamma$ and satisfies $f\left( x \right) = \left\{ {\begin{array}{*{20}{c}}
0&{x \le r}\\
1&{x > r}
\end{array}} \right.$.
  \end{definition}
  Four situations could occur on the $i$-th pilot subcarrier: 1) No signal exists, that is, ${{\bf{y}}_i}\left[ k \right] = {{\bf{w}}_i}$; 2) Only signals from Eva exist, that is, ${{\bf{y}}_i}\left[ k \right] = {\bf{g}}_i^{\rm{E}}{n_i}\left[ k \right] + {{\bf{w}}_i}$; 3) Only signals from LUs exists, that is,  ${{\bf{y}}_i}\left[ k \right] = \sum\limits_{j = 1}^{{{M_{i}}}} {{\bf{g}}_{j,i}^{\rm{L}}{x_{{\rm{L}},j}}\left[ k \right] + } {{\bf{w}}_i},{{M_{i}}} \ge 1$; 4)  Both of signals from LUs and Eva exist, that is, ${{\bf{y}}_i}\left[ k \right] = {\sum\limits_{j = 1}^{{{M_{i}}} - 1} {{\bf{g}}_{j,i}^{\rm{L}}{x_{{\rm{L}},j}}\left[ k \right] + {\bf{g}}_i^{\rm{E}}{n_i}\left[ k \right] + } {{\bf{w}}_i},{{M_{i}}} \ge 2}$.

To identify the principle of designing $\gamma$,  we give the following interpretation.  According to   law of large numbers,   the inner product  between  signals from two independent  individuals, i.e. a LU and Ava,  approaches  zero~\cite{Hoydis}. On the contrary, the inner product between signals from the same  node can reach  a value with its amplitude equal to one. In theory,  the value of $r$ can thus be configured to be a certain value, i.e., 0.5.
\begin{remark}
The diversity of column-wise cover-free coding is a pure code property, independent  with  the characteristic  of per-word independence-aided differential coding which is intrinsically  a data-driven concept.   We stress that the two codes  coexist without affecting each other, and naturally agree with each other when the number of antennas increases.
\end{remark}
 In what follows, we will prove that the  coding diversity  can perfectly provide multiuser guarantee but  lack the ability of resisting attack without the assistance of the row property which has to be exploited  in the decoding part. Nevertheless, we stress that the coding diversity does not make no sense. The row property has been coupled  with and included in  the code but not yet been  exploited in the encoding procedure.
\begin{fact}[\textbf{Achieving Multiuser Guarantee of Procedure 2}]
We in this paper consider the special case where $L=1$. The cover-free coding is introduced in~\cite{Kautz,yachkov}. In this case, the boolean sum of any subset of $k \le K$ codewords in $\bf B$  does not cover any codeword that is not in the subset.  For a constant-weight H2DF-$\left( {K,1,B} \right)$ code,  two arbitrary  SP sums,  each superimposed by $k \le K$ codewords,  are identical if and only if the two  codeword sets  constituting the two sums are completely identical as well.  This property  guarantees the UD  property.

We spilt the columns  of H2DF-$\left( {K,1,B} \right)$ code  matrix $\bf{B}$ into $K$ independent clusters. The $i$-th  cluster  includes $\left[ {{C \mathord{\left/
 {\vphantom {C K}} \right.
 \kern-\nulldelimiterspace} K}} \right]$  columns indexed by ${\cal B}_{i}$ and  constitutes a so-called \textbf{submatrix} denoted  by   $\left[ {{{\bf{b}}_{j \in {{\cal B}_{i}}}}} \right]$ which is  exclusively allocated  to the $i$-th LU.  Since  the UD property of $\bf{B}$  has been satisfied, those submatrices can naturally satisfy the properties of  UPrSI and UPoSI. In this way,  H2DF-$\left( {K,1,B} \right)$ code  can support multiuser guarantee, provided that there is no attack.
\end{fact}

\begin{definition}
A $B\times C$ binary matrix $\bf M$  is called a paired H2DF-$\left( {K,1,B} \right)$ code of length $B$, size $C$ and order $K$,
if $\bf M$, equal to  $\bf B$, has its codewords obeying  the ASP principle
\end{definition}
\begin{definition}
 We define the ${B \times \left( {\begin{array}{*{20}{c}}
C\\
k
\end{array}} \right)}$ SP-sum matrix  ${{\bf{B}}_k}, k = 2, 3,..., K$, which is the collection of all of the SP sums of  codewords from $\bf B$, taken exactly $k$ at a time.   Correspondingly, the ${B \times \left( {\begin{array}{*{20}{c}}
C\\
k
\end{array}} \right)}$ ASP-sum matrix ${{\bf{M}}_k}, k = 2, 3,..., K$ denotes the collection of all of the ASP sums of those codewords collected from $\bf B$,  taken exactly $k$ at a time. Each column  of ${{\bf{B}}_k}$ (or ${{\bf{M}}_k}$)  represents a unique SP-sum (or ASP-sum) codeword.
\end{definition}
Let us examine the ability of column property  of H2DF encoding principle  in resisting attack.  Ava could use its intended SAPs to cause confusion on the detection of  signals  on  any victim  subcarrier.  For example, we consider  cover-free  codewords of three LUs, that is, $ \left[ {\begin{array}{*{20}{c}}
1&0&0 &1
\end{array}} \right]$, $ \left[ {\begin{array}{*{20}{c}}
0&0&1 &1
\end{array}} \right]$, and $ \left[ {\begin{array}{*{20}{c}}
1&0&1 &0
\end{array}} \right]$ and an  codeword from Eva, that is, ${\bf{c}} =\left[ {\begin{array}{*{20}{c}}
1&1&0 &0
\end{array}} \right]$. After detection, Alice could derive the final superposed binary codeword  equal to $\left[ {\begin{array}{*{20}{c}}
1&1&1 &1
\end{array}} \right]$, which actually  indicates no any useful information and imposes huge confusion  on cover-free decoding.
 In this case, the decoding process is imprecise and  multiuser guarantee is paralyzed.

\section{Hierarchical 2-D Feature Coding for SMPA Architecture: Decoding  Part}
\label{H2DFCD}
We in this section attempt to build up the \textbf{Attack Resistance} of Procedure 2  while maintaining the  perfect multiuser guarantee. Basically, we aim to reach the potential, that is,
 \emph{On the basis of  the energy feature, Alice  could take advantage of the independence feature  to recover the  significant potential of coding diversity for securing against the hybrid attack.} This is done by upgrading  the decoding part of  UD property through a  HD  model shown in Fig.~\ref{MLR}. In this model,  operation  units with specific functionalities  are connected and distributed among  five  sub-layers, including  the starting layer, input layer, hidden layer, supervision layer and output layer. In what follows, we begin our discussion by the starting layer.
\subsection{Starting Layer (\textbf{Signal Collection $\&$ Mathematical Modeling for Signal Superposition in SMPA Architecture})}
 \textbf{Input:} Multiple independent  signals from $K+1$ nodes.

As the start of SMPA architecture,  the superposition of   observed signals on subcarriers at Alice   brings the superposition of SAPs and thus  the superposition of codewords mathematically.  Two types of  codeword superposition  can be  characterised by:
\begin{equation}
{{\bf{b}}_1} \bigvee  \cdots  \bigvee {{\bf{b}}_K}={{\bf{b}}_{{\rm{S}},K}}, {{\bf{b}}_{{\rm{S}},K}} \bigvee {\bf{c}}={{\bf{b}}_{\rm{I}}}
\end{equation}
and
 \begin{figure*}[!t]
\includegraphics[width=1\linewidth]{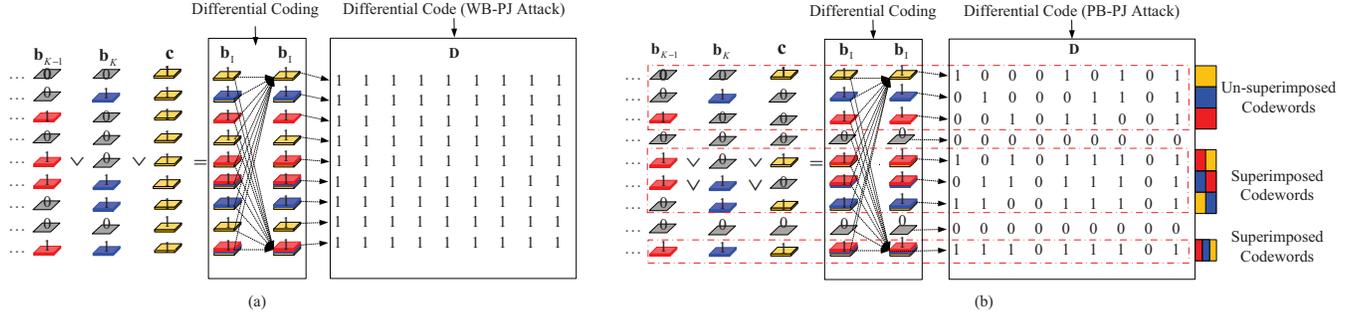}
\caption{ Examples of 2-D differential coding; (a) Under WB-PJ attack; (b) Under PB-PJ attack. }\label{H2DF-D-coding1}
\end{figure*}
\begin{equation}
{{\bf{m}}_1} +  \cdots  + {{\bf{m}}_K}= {{\bf{m}}_{{\rm{S}},K}}, {{\bf{m}}_{{\rm{S}},K}} + {\bf{c}}={{\bf{m}}_{\rm{I}}}
\end{equation}
where  ${{\bf{b}}_i}\in \left[ {{{\bf{b}}_{j{'} \in {{\cal B}_i}}}} \right]$ and ${{\bf{b}}_i}={{\bf{m}}_i},\forall i, 1\le i\le K$.
 Here, the specific value of vector ${\bf{c}}$,  determined by   the attack behaviors, is defined by:
\begin{equation}
{\bf{c}} = \left\{ {\begin{array}{*{20}{c}}
{{{\left[ {\begin{array}{*{20}{c}}
0& \cdots &0
\end{array}} \right]}^{\rm{T}}}}&{ \rm SC}\\
{{{\left[ {\begin{array}{*{20}{c}}
1& \cdots &1
\end{array}} \right]}^{\rm{T}}}}&{\rm WB\text{-}PJ}\\
{{{\left[ {\begin{array}{*{20}{c}}
0& \cdots &1
\end{array}} \right]}^{\rm{T}}}}&{\rm PB\text{-}PJ}
\end{array}} \right.
\end{equation}

However, the mathematical hints above can not obscure the fact that what Alice actually receives on subcarriers are still superimposed signals, rather than the code digits.

\textbf{Output:}  The superimposed signals which are  stored  to the  data units in the  next layer.

\subsection{Input Layer (\textbf{Feature Extraction})}
\textbf{Input:} The  superimposed signals from the previous layer.

Alice  aims to  extract   features of  superimposed signals  in data units, and encode them  into code digits.
Depending on features involved, those superimposed signals undergo two specific independent processes~(See Fig.~\ref{MLR}), that is, energy feature extraction and independence feature extraction.
 \subsubsection{Energy Feature Extracting by ERED} The  detection units, configured in  this layer,   extract energy feature from  superimposed signals in the form of  codewords, i. e., ${{\bf{b}}_{\rm{I}}}$ and ${{\bf{m}}_{\rm{I}}}$. This process is same with the one shown in Section~\ref{EFE}.  ${{\bf{b}}_{\rm{I}}}$ and ${{\bf{m}}_{\rm{I}}}$   will be  delivered to the digit units configured in the next layer.
 \subsubsection{Independence Feature Extracting by Inner Product Operation}

 The inner product units are configured to extract  the independence feature from  the superimposed signals in the form of code matrix, namely, $\bf D$.  See details in Algorithm~\ref{Alg1}.
 \begin{algorithm}[!t]
\caption{Formulation of  2-D Differential Code Matrix}
\begin{algorithmic}[1]
 \label{Alg1}
\FOR {$i=1$ to $i=B$}
\STATE Select  superimposed signals at  the  $i$-th subcarrier as the reference  input of inner product unit;  Use ERED to determine the $i$-th binary digit.
\STATE Input a total of $B$  superimposed signals  as the other  input of the same inner product unit.
\STATE The  inner product unit  perform inner product operation between the two inputs.   Differential binary code digits of length $B$  are formulated.
\STATE  Perform XOR operation between the $j$-th differential code digit and the $i$-th reference digit (usually 1 if signals exist, otherwise 0.). The result is  $d_{i,j}$.
\ENDFOR
\STATE Because a total of  $B$ reference digits can be specified precisely, a total of $B$ codewords of length $B$  is formulated,  constituting a 2-D differential code matrix defined in Definition 8.
\end{algorithmic}
\end{algorithm}
\begin{definition}
 A 2-D Differential code  is defined by a $B \times B$ matrix ${\bf{D}} = \left[ {{{\bf{d}}_{j \in \left[ {1,B} \right]}}} \right] $ with its $j$-th row ${{\bf{d}}_j} $ denoted by ${{\bf{d}}_j} = {\left[ {\begin{array}{*{20}{c}}
{{d_{1,j}}}& \cdots &{{d_{B,j}}}
\end{array}} \right]}$.
\end{definition}
Thanks to the feature extraction,  the derived $\bf D$ matrix has the potential of including  all the information of codewords employed by  LUs and/or Ava. $\bf D$  can contribute  to the decomposition of  ${{\bf{b}}_{\rm{I}}}$ in the sense that it can facilitate  the precise detection of $\bf{c}$  which is then eliminated from  ${{\bf{b}}_{\rm{I}}}$.

\textbf{Output:}${{\bf{b}}_{\rm{I}}}$,  ${{\bf{m}}_{\rm{I}}}$, and $\bf D$. Those are delivered to the digit units in the next layer.
 \subsection{ Codeword Superposition and Uncertainty Principle}

 It should be stressly noted that  \emph{unpredictable attack behaviors and the corresponding codewords as well as  random utilization of codewords of LUs} could prevent Alice from acquiring  the  specific codewords within  $\bf D$. The direct result is that  Alice will observe that superimposed codewords coexist with the  un-superimposed ones.

 This prompts us  to  discover and formulate the concept of CSUP, a basic and deterministic  rule followed by those random codewords within  $\bf D$  under various attacks. Undoubtedly, CSUP will contribute to the clear understanding  of $\bf D$,  fundamentally facilitating the decomposition of ${{\bf{b}}_{\rm{I}}}$ in the following layers.

  In order to   uncover  the secret of CSUP,  we need to examine  the codeword  superposition process in a smart, explicit and institutive  way. Let us focus on a special stacked matrix $\bf T$.
 \begin{definition}
Imagine  a $\left( {K{\rm{ + 2}}} \right)$-by-$B$ matrix ${\bf T}=\left[ {{t_{j,i}}} \right]$ composed by ${\bf{T}} = {\left[ {\begin{array}{*{20}{c}}
{{{\bf{b}}_{\rm{1}}}}& \cdots &{{{\bf{b}}_K}}&{\bf{c}}&{{{\bf{b}}_{\rm{I}}}}
\end{array}} \right]^{\rm{T}}}$ where ${{\bf{b}}_{{i}}}$ and ${{\bf{b}}_{\rm{I}}}$  satisfy Eq. (11). The  $i$-th column of the submatrix  constituted by the first  $K$ row  is denoted by ${{\bf{t}}_i}$ with  ${{\bf{t}}_i} = {\left[ {\begin{array}{*{20}{c}}
{{t_{1,i}}}& \cdots &{{t_{K,i}}}
\end{array}} \right]^{\rm{T}}}$.  The sum of  elements of ${{\bf{t}}_i}$ is defined by ${t_{s,i}}$ with  ${t_{s,i}} = \sum\limits_{j = 1}^K {{t_{j,i}}} $. The column index  of ${\bf{T}}^{\rm T}$ corresponds to that of $\bf D$, or equivalently the index of  reference digits.
 \end{definition}
 We cannot know  the identities of codewords within ${\bf{T}}$ exactly. In fact,  we may not care the exact identities of codewords,  but instead  concentrate our attention on the column and row property of ${\bf{T}}$ from the following two views:
\begin{fact}[\textbf{CSUP: Un-superimposed Codewords}]
 If ${t_{s,i_{0}}}+t_{K+1,i_{0}}=1$  holds true, or equivalently, ${{{m}}_{{\rm{I}},{i_{0}}}}=1$ holds true,   Alice is perfectly able to  deduce that  the recovered  codeword at the ${\bf d}_{i_0}$, also namely, the \textbf{exposed codeword},  must be un-superimposed.
\end{fact}

\begin{fact}[\textbf{CSUP: Superimposed Codewords}]
For any column $i_{0}$ satisfying $t_{s,i_{0}}= m \ge  2$,   a total of $m$ codewords  are surely superimposed together at this  column.  Only the superposition version of  $m$ codewords is enabled  to be precisely recovered (or namely exposed)  and  equal to   ${\bf d}_{i_0}$.
\end{fact}
CSUP describes the potential  structures of  $\bf D$  under hybrid  attacks comprehensively.   Two typical examples demonstrating  CSUP could be respectively  depicted in Fig.~\ref{H2DF-D-coding1}.  Keep in mind that CSUP provides  a very basic background and technical support for solving the issues in remainder of subsections. It causes us to direct  our attention to the exposed codewords which are however instable to appear.  Therefore, we could sufficiently believe that the following process of codeword classification and identification (CI) could be seriously destabilized.  This calls for more advanced and delicate mechanism to further optimize the decoding process, which will be depicted in the following layers.
\subsection{Hidden Layer  (\textbf{Attack Identification})}
Backing to the hidden layer, we  stress that  this layer is different from the previous layers and begins to gradually resolve the attack related  issues.

 \textbf{Input:} $\bf D$,   ${{\bf{b}}_{\rm{I}}}$,  ${{\bf{m}}_{\rm{I}}}$, ${{\bf{B}}_{K}}$ and ${{\bf{M}}_{K}}$ in the storage units.

 As shown in Fig.~\ref{MLR}, those inputs  are fed to the identification units which are responsible for two tasks, that is,  1)  precisely identify the current  attack mode (SC, WB-PJ or PB-PJ); 2) differentiate the observed codeword  in the current  mode from  those codewords under potential interfering  modes.
\subsubsection{{Identification of  WB-PJ  Attack}}
CSUP tells us that  the codewords under WB-PJ attack is very distinctive compared with those under other attacks. The principle can be summarized as follows:
\begin{proposition}
  When zero digit does not exist  in the inputs of ${{\bf{b}}_{\rm{I}}}$  and $\bf{D}$, WB-PJ attack happens. This unique characteristic  differentiates  WB-PJ attack from other two attacks, i.e., SC and PB-PJ attack.
\end{proposition}
\begin{algorithm}[!t]
\caption{Codeword CI Under WB-PJ and SC Attack}
\begin{algorithmic}[1]
\label{Algg2}
\STATE (For WB-PJ attack) Derive a novel paired H2DF codeword  by performing the ASP sum between  ${{\bf{m}}_{\rm{I}}}$  and  a vector with all digit given``-1'' value;
\STATE Compare this codeword with each column of matrix ${{\bf{M}}_K}$ and find the index of the column equal to this codeword exactly;
\STATE  Back to the matrix ${{\bf{B}}_K}$ and find the superimposed codeword ${{\bf{b}}_{{\rm{S}},K}}$ at the same column index.
\STATE (For both WB-PJ attack and SC) Decompose  ${{\bf{b}}_{{\rm{S}},K}}$  into  ${{\bf{b}}_{i}},\forall i, 1\le i \le K$ based on the column-wise cover free coding characteristic.
\end{algorithmic}
\end{algorithm}
\subsubsection{Identification of  SC and PB-PJ  Attack}
 Theoretically, we need to examine the existence of  $\bf{c}$ and  further clarify the superposition relationship between ${{\bf{b}}_{{\rm{S}},K}}$ and  $\bf{c}$ such that we can directly identify SC and PB-PJ  attack.  However,  what Alice observes  in practice is  the their superposition version, i.e., ${{\bf{b}}_{\rm{I}}}$ and ${{\bf{m}}_{\rm{I}}}$, rather than ${{\bf{b}}_{{\rm{S}},K}}$ and  $\bf{c}$.  We stress that ${{\bf{B}}_{K}}$ and ${{\bf{M}}_{K}}$ make the identification of attack more easily.
\begin{proposition}
When  ${{\bf{b}}_{\rm{I}}}\notin {\overline{\cal{B}}}_{K}$  holds true, there exists ${{\bf{b}}_{\rm{I}}}\ne{{\bf{b}}_{{\rm{S}},K}}$, which means that $\bf{c}$  cannot be covered by ${{\bf{b}}_{{\rm{S}},K}}$.  In this case,  there   must exist one column indexed by $i_{a}$, $i_{a} \in {\cal D}$  such that $c_{i_{a}}$, i.e., $t_{K+1,i_{a}}$ in $\bf T$, is equal to 1,  with  ${t_{s,i_{a}}} =0 $.  The attack is classified as  PB-PJ attack or WB-PJ attack. We can  surely differentiate   the PB-PJ attack from WB-PJ attack where   all the inputs of $\cal D$  are uniquely non-zero digits.
\end{proposition}
The proof is easy since  if otherwise  there does not exist ${t_{s,i}} = 0 $ \,$\forall i$,   then  $\bf{c}$ must  be covered by ${{\bf{b}}_{{\rm{S}},K}}$ and   ${{\bf{b}}_{\rm{I}}}={{\bf{b}}_{{\rm{S}},K}}$.
\begin{proposition}
When  ${{\bf{b}}_{\rm{I}}}\in{\overline{\cal{B}}}_{K}$ holds true, there exists ${{\bf{b}}_{\rm{I}}}={{\bf{b}}_{{\rm{S}},K}}$ and $\bf{c}$  is covered by ${{\bf{b}}_{{\rm{S}},K}}$. On this basis, SC occurs if ${{\bf{m}}_{\rm{I}}} ={{\bf{m}}_{{\rm{S,}}K}}$, and otherwise,  PB-PJ attack happens.
\end{proposition}
In practice, we should note that the above theoretical results  depend on  how to  determine  the relationship between  ${{\bf{b}}_{\rm{I}}}$ ( ${{\bf{m}}_{\rm{I}}}$ ) and  ${{\bf{B}}_K}$ (${{\bf{M}}_K}$).  In order to achieve this,  Alice can search ${{\bf{B}}_K}$ (${{\bf{M}}_K}$), find whether ${{\bf{b}}_{\rm{I}}}\in{\overline{\cal{B}}}_{K}$ or ${{\bf{m}}_{\rm{I}}}\in{\overline{\cal{M}}}_{K}$ holds true  and further  make  decisions shown in above theorems.
The overall algorithm  can be shown in Fig.~\ref{figure-Classification}.

\textbf{Output:} Identified attack mode and the superimposed codeword under the identified mode.

\begin{figure}[!t]
\centering \vspace{-10pt}\includegraphics[width=0.75\linewidth]{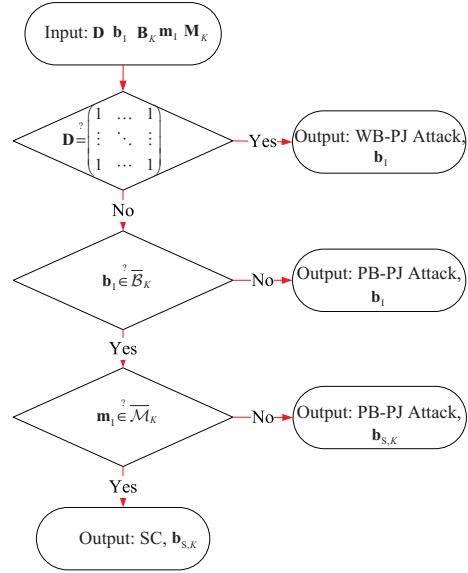}
\caption{Flow chart of attack identification.}
\label{figure-Classification}
\end{figure}

\begin{algorithm}[!t]
\caption{BFPI Algorithm, also see Fig.~\ref{FCDS}}
\begin{algorithmic}[1]
\label{Alg_3}
\STATE Formulate two hypotheses as follows:
\begin{equation}
\begin{array}{l}
{{\cal H}_0}: {\bf b}_{K,{j}_{0}}\to {\rm LUs},{\bf b}_{K,{j}_{1}} \to {\rm Ava}\\
{{\cal H}_1}:{\bf b}_{K,{j}_{1}}\to {\rm LUs},{\bf b}_{K,{j}_{0}} \to {\rm Ava}
\end{array}
\end{equation}
\STATE (\textbf{Backward Propagation (BP)}) For ${{\cal H}_0}$,   ${\bf b}_{K,{j}_{0}}$  and ${\bf b}_{K,{j}_{1}}$ can both  be decomposed into original codewords in  $\bf B$, i.e., ${\left\{ {{{\overline {\bf{b}} }_{i,0}}} \right\}_{i = 1, \ldots K}}$ and ${\left\{ {{{\overline {\bf{b}} }_{i,1}}} \right\}_{i = 1, \ldots K}}$.
\STATE (\textbf{Forward Propagation (FP)})  According to the definition of H2DF code, any two different superimposed codewords can be decomposed into  two different codeword sets.  Therefore, the ASP sum must be different.  Calculate the ASP sum  as ${ {\overline{ \bf{b}}} _{{\rm{S}},0}} = {{\overline {\bf{b}}} _{1,0}} +  \cdots  + { {\overline {\bf{b}}} _{K,0}}+{\bf b}_{K,{j}_{1}}$.
\STATE  (\textbf{Decision})   If  ${{\overline{\bf{b}}} _{{\rm{S}},0}}$ is equal to the observation, then we know that  ${{\cal H}_0}$ is valid, otherwise, ${{\cal H}_1}$ is valid.
\end{algorithmic}
\end{algorithm}
\subsection{Supervision Layer (\textbf{Codeword CI})}
\textbf{Input:}  the superimposed codeword, i.e.,  ${\bf{b}}_{\rm I}$ or ${\bf{b}}_{{\rm S}, K}$, under the identified attack mode,   ${{\bf{B}}_K}$ and ${{\bf{M}}_K}$  in  the storage units.

The CI units are responsible for  decomposing ${\bf{b}}_{\rm I}$ or ${\bf{b}}_{{\rm S}, K}$, and  executing the following CI  task for the  derived codewords.

For WB-PJ and SC attack,  decomposing ${\bf{b}}_{\rm I}$ or ${\bf{b}}_{{\rm S}, K}$ could be done in Algorithm~\ref{Algg2}.

For PB-PJ attack, the flexible choices of  $\bf{c} $ could cause two  results. 1)~When  ${{\bf{b}}_{\rm{I}}}={{\bf{b}}_{{\rm{S}},K}}$,  $\bf{c} $ is  covered by ${{\bf{b}}_{{\rm{S}},K}}$. This issue  of codeword decomposition and codeword CI  could be  resolved using Algorithm~\ref{Algg2}. 2)~Otherwise when ${{\bf{b}}_{\rm{I}}}\ne{{\bf{b}}_{{\rm{S}},K}}$, $\bf{c} $ is not covered by ${{\bf{b}}_{{\rm{S}},K}}$. The situation is rather  complicate.

In principle, the precise decomposition of  ${{\bf{b}}_{\rm{I}}}$ into ${{\bf{b}}_{{\rm{S}},K}}$ is guaranteed iff  the exposed codeword  ${{{\bf d}}_{{i}\in {\cal D}}}$ (See its definition in Fact 2) is precisely identified as $\bf{c}$ which is then  eliminated exactly from ${{\bf{b}}_{\rm{I}}}$.
The challenge lies in how to identify the identities of  ${{{\bf d}}_{{i}\in {\cal D}}}$ which are usually indeterministic due to the random utilization of codewords.
\begin{problem}[Indeterministic Exposed Codewords]
When PB-PJ attack happens and  ${{\bf{b}}_{\rm{I}}}\ne{{\bf{b}}_{{\rm{S}},K}}$ exists,  the identities of  elements in  $ {\cal D}$ are random and unpredictable   due to the random codewords employed by  the overall  $K+1$ nodes.   Therefore the identities of ${{{\bf d}}_{{i}\in {\cal D}}}$  are  unknown.  Any prior constraints that are imposed on the set for resolving this issue will also  break down the randomness of original codewords and must  not be considered.
\end{problem}
 \begin{figure}[!t]
\centering \includegraphics[width=1.0\linewidth]{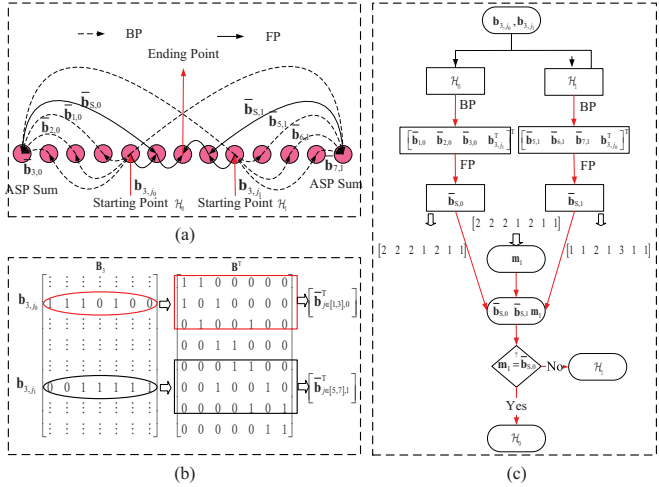}
\caption{Example  of BFPI Algorithm in  the supervision layer, for identifying the confused codewords ${\bf b}_{K,{j}_{0}}$  and ${\bf b}_{K,{j}_{1}}$.  (a) Graph  of CI units  under  BFPI Algorithm  for PB-PJ attack; (b) Codeword decomposition for the confused codewords ${\bf b}_{K,{j}_{0}}$  and ${\bf b}_{K,{j}_{1}}$ ; (c) Flow chart of  BFPI Algorithm.
.}\label{FCDS}
\end{figure}
Institutively,  the precise decomposition of  ${{\bf{b}}_{\rm{I}}}$ into ${{\bf{b}}_{{\rm{S}},K}}$  is out of the question.
However,  if we  exhaust all possibilities  of $\bf{c}$  which actually could be any codeword, we find that  the codeword CI is   intrinsically hindered by two types of confusions.

1) The identification confusion caused by $\bf{c}$ when ${{\bf{c}}}\in{\overline{\cal{B}}}_{K}$ holds true. In this case, there must exist  $i_{a}$ such that ${\bf{d}}_{i_{a}}=\bf{c}$. Besides, we know  that ${{\bf{b}}_{{\rm{S}},K}}\in{\overline{\cal{B}}}_{K}$ always holds true. Therefore, there exists the following problem:
\begin{problem}
 A superposition identification (SPI)  problem  happens, that is,  $\bf{c}$  cannot  be differentiated from  ${{\bf{b}}_{{\rm{S}},K}}$ since they are both in the same matrix ${\bf B}_{K}$.
\end{problem}
Note that the case where  $\bf{c}={{\bf{b}}_{{\rm{S}},K}}$ does not affect the ultimate decoding of codewords. Without loss of generality, we define  the exact occurrence   of $\bf{c}$ and ${{\bf{b}}_{{\rm{S}},K}}$   in ${\bf B}_{K}$   by ${{\bf{b}}_{{{K}},j_{0}}}$ and ${{\bf{b}}_{{{K}},j_{1}}}$. The identities of ${{\bf{b}}_{{{K}},j_{0}}}$ and ${{\bf{b}}_{{{K}},j_{1}}}$  are both unknown in practice. To completely resolve the above  problem, we develop the technique of \textbf{Back/Forward Propagation Identification (BFPI)}. The details  are given in Algorithm~\ref{Alg_3} and an example of BFPI for 7-digit codewords under $K=3$ LUs  can be seen in the Fig.~\ref{FCDS}.
\begin{algorithm}[!t]
\caption{ Codeword CI}
 \begin{algorithmic}[1]
 \label{Alg_6}
\STATE Search the set ${{\cal D}}$ and derive  ${{\bf{d}}_{{i}}}$,  ${i} \in {\cal D}$. Those codewords include the one from Eva and those from LUs.
\STATE  Alice  select an exposed  digit with  position ${i_d} \in {\cal D}$ and then configure the ${i_d}$-th digit of  ${{\bf{b}}_{\rm I}}$ to be zero. ${{\bf{b}}_{i}},\forall i, 1\le i \le K$   is  recovered by  searching ${\bf B}_{K}$ for the codeword  identical to the revised  ${{\bf{b}}_{\rm I}}$.
\STATE Alice examines  the  overall distribution of ${{\bf{b}}_{i}},\forall i, 1\le i \le K$ in $\bf B$ and  determines the selected digit  ${i_d}$ belonging to the true  codeword of Ava only when the  distribution satisfies MUCD.  According to this  digit, then ${\bf c}={{\bf{d}}_{{i_d}}}$. Then  Alice can finally confirm ${{\bf{b}}_{i}},\forall i, 1\le i \le K$ in this case to be the right codewords  from LUs.
 \end{algorithmic}
\end{algorithm}

2) The identification confusion caused by $\bf{c}$ when $\bf{c}$  belongs to the codebook $\bf B$, that is, ${\bf{c}}\in \bigcup\limits_{i = 1}^K {{{{\cal B}}_i}} $.  More specifically, $\bf{c}$ could be located in  any  uncertain submatrix, or namely,  $\bf{c}$  contaminates one  submatrix. There exists ${\bf{c}}  \in {\left\{ {{{ {\bf{d}} }_{i}}} \right\}_{i \in \cal D}}$.
  \begin{definition}[Multi-User Codeword Distribution (MUCD)]
 MUCD means that  there always exists a unique codeword  in use for  each submatrix $\left[ {{{\bf{b}}_{j \in {{\cal B}_{i}}}}} \right]$.
\end{definition}
  However, the indeterministic  relationship between ${\bf{c}}$  and exposed codeword ${\bf d}_{i}$  causes  a random disturbance (RD) problem in the recovery operation
\begin{problem}
\textbf{Confusing case:} when  there exist $j_{1}$ and a set ${\cal D}_{0}\subseteq {\cal D}, {\cal D}_{0}\ne \emptyset $ such that  $\forall i, i \in {\cal D}_{0}$, ${\bf d}_{i}$ and ${\bf{c}}$ are located  within the same submatrix $\left[ {{{\bf{b}}_{j \in {{\cal B}_{j_{1}}}}}} \right]$, $\bf{c}$  could not be differentiated from ${\bf d}_{i\in {\cal D}_{0}}$. \textbf{Identifiable case:} Otherwise, $\bf{c}$  could be differentiated from ${\bf d}_{i\in {\cal D}_{0}}$ using Algorithm 4. This is done by  examining whether or not   the recovered MUCD  is true.    However, Alice could not predict the occurrence frequency of  two cases since the exposed codewords are random, which reduces the reliability of  pilot encoding/decoding in this architecture.
\end{problem}
RD problem  causes   an  instable  codeword CI  for LUs. We will analyze  this phenomenon  in the next section since in this section we only focus on the design of architecture.

\textbf{Output:} Precise ${{\bf{b}}_{i}},\forall i, 1\le i \le K$ under SC and WB-PJ attack; Unstably identified  ${{\bf{b}}_{i}},\forall i, 1\le i \le K$ for PB-PJ attack,

\subsection{Output Layer}
This layer is configured for storing  the  codewords classified and identified  from the previous layer.

\section{Reliability Bound Contraction Theory}
\label{RBCT}
In this section, we exploit the concept of IEP to measure  the  reliability of  SMPA. But, our main work is to mathematically characterize the instability  in this reliability and then aim to answer two questions, that is, how  to   reduce   the  instability and what level of stability can be achieved. This will be done by our proposed  RBC Theory.
\subsection{PA Reliability and Its Instability}
 In this hybrid attack scenario,  the codeword identification error  occurs  due to the PB-PJ attack, rather than WB-PJ and SC attack.  This could be easily proved using Proposition  1, 2 and 3.
\begin{theorem}
Given $K$ users and $ N^{\rm L}_{\rm P}$ subcarriers,  the IEP which is denoted  by $P$  under H2DF-$\left( {K,1,B} \right)$ code with size $C$ is  bounded by:
\begin{equation}\label{E.41}
P_{\rm lower} \le P \le P_{\rm upper}
\end{equation}
where $P_{\rm lower} = \frac{1}{C}$ and $P_{\rm upper}=\frac{1}{{{2K}}}$. The reliability of  SMPA is  defined by
\begin{equation}
{R_S} =  - {\log _{10}}P
\end{equation}
\end{theorem}
\begin{proof}
 We assume that Eva is interested in the $i_{0}$-th LU and chooses a random codeword from submatrix $\left[ {{{\bf{b}}_{j \in {{\cal B}_{i_{0}}}}}} \right]$ as $\bf{c}$. What is observed at Alice is that the exposed codewords  from LUs could be randomly located  in arbitrary one of $K$ possible submatrices  and independent  with  the codeword choice of Eva. The  worst case happens if the confusing case occurs.  The probability is equal to ${{\rm{1}} \mathord{\left/
 {\vphantom {{\rm{1}} {{K}}}} \right.
 \kern-\nulldelimiterspace} {{K}}}$ and in this case,  the right  identification happens with probability ${{\rm{1}} \mathord{\left/
 {\vphantom {{\rm{1}} {{2}}}} \right.
 \kern-\nulldelimiterspace} {{2}}}$. The final IEP  is calculated as ${{\rm{1}} \mathord{\left/
 {\vphantom {{\rm{1}} {{2K}}}} \right.
 \kern-\nulldelimiterspace} {{2K}}}$.

 Otherwise, a best case (i.e., the identifiable case) occurs. The  IEP is then  transformed into  the probability of the occurrence of duplicate codewords  among  the decomposed codewords ${{\bf{b}}_{i}},\forall i, 1\le i \le K$, and  thus  calculated as  ${{\rm{1}} \mathord{\left/
 {\vphantom {{\rm{1}} C}} \right.
 \kern-\nulldelimiterspace} C}$.
\end{proof}
The non-tight IEP bounds  tell us that  the exact evaluation of  SMPA reliability depends on  the realization of  SMPA. From a long-term perspective, we can think that the SMPA reliability fluctuates in time.  Everytime SMPA is run, the differing random input (codewords) leads to a different random output, or the realization, of the SMPA process. The realization  of  SMPA process denoted by $X$ for the outcome IEP $P$ is the function $X\left( {t,P} \right)$, defined by $t \mapsto X\left( {t,P} \right)$.  However, how to model the statistic process is  our focus in the future and instead we hope to find a easy-to-implement technique to avoid this uncertainty smartly even though it means sacrificing some performance.

To pursue the matter further we define the maximum IEP difference between  arbitrary two realizations of SMPA process within  all possible time slices, as the  long-term instability. More specifically, we have:
\begin{definition}
The long-term instability in SMPA reliability  is defined by:
\begin{equation}
{S_R} = {\log _{10}}\left( {{{{P_{{\rm{upper}}}}} \mathord{\left/
 {\vphantom {{{P_{{\rm{upper}}}}} {{P_{{\rm{Iower}}}}}}} \right.
 \kern-\nulldelimiterspace} {{P_{{\rm{Iower}}}}}}} \right)
\end{equation}
\end{definition}
Basically, precise repeated assessment  is very critical for  SMPA as we do not hope to encounter  a situation where every time the system operator uses it, the evaluation  of its reliability is provided  imprecisely. Therefore, the proposed technique   should be able to reduce ${S_R} $ to zero.

\begin{figure}[!t]
\centering \includegraphics[width=1\linewidth]{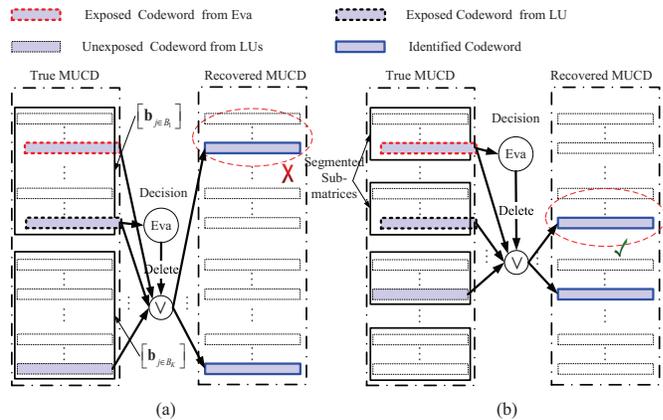}
\caption{ Identification enhancement for the confusing case.
For instance, Alice identifies the sub-codebook $\left[ {{{\bf{b}}_{j \in {{\cal B}_1}}}} \right]$ as the contaminated codebook. (a) The wrong decision happens because assuming any one of two exposed codewords in $\left[ {{{\bf{b}}_{j \in {{\cal B}_1}}}} \right]$  as $\bf{c}$ which is further eliminated from ${{\bf{b}}_{\rm I}}$  will generate true  MUCD; (b) Transform the wrong decision to the right decision with the help of segmented submatrices  and MUCD. }
\label{Sub-C3}
\end{figure}
\subsection{Observation on the Exposed Codewords}
Backing to the RD problem, we find that the key of reducing  ${S_R} $ lies in how to reduce the occurrence frequency of  \textbf{Confusing Case}.  Basically, it requires  that Alice is able to \emph{more precisely}  locate  the scope of the position fluctuation of exposed codewords  in  $\bf B$  and further differentiate  between different results. The most important is to  discover the  controllable variable achieving this.

We find that  the submatrix-level  resolution  is low. We bring up the subject of   the low submatrix-level  resolution here as each LU is \emph{solely} assigned with one submatrix and there exists only  $K$ choices in total for each LU.  In this sense,  the low resolution makes  the  scope of the position fluctuation of exposed codewords relatively \emph{extensive}. Therefore, everytime LUs select  their own codewords obeying  MUCD and Eva employs $\bf{c} $ which is identical to the codeword within the same submatrix as one  LU of interest,   the  \emph{ probability} that  the set of exposed codewords includes  $\bf{c} $   is high.  This can be seen in  Fig.~\ref{Sub-C3} (a).

However, an interesting phenomenon is that the codeword-level resolution is very high, namely,  a huge number of candidate  codewords for each LU exist.
\subsection{RBC Theory: Code Partition and Upper-Lower Bound Tradeoff}
The above observation  inspires us to perform  codebook partition for each submatrix  $\left[ {{{\bf{b}}_{j \in {{\cal B}_{i}}}}} \right], 1\le i \le K$, in other words,  the  controllable variable, denoted by $N$,  now is  identified as the number of \textbf{segmented submatrices}.  Theoretically,  if each of  LUs is assigned with arbitrarily  one  of  segmented submatrices, the codeword-level resolution is reduced and the lower bound of  IEP is enlarged. Fortunately, the submatrix-level  resolution  will be thus  improved and the scope of the position fluctuation of exposed codewords  is restricted. Fig.~\ref{Sub-C3}(b) shows how the recovered MUCD  in  confusing case is able to be transformed to be  identifiable. Then the upper bound  is reduced.
\begin{fact}[Upper-Lower Bound Tradeoff]
On  one hand, the larger $N$ is,  the lower  the upper bound is.  On the other hand, the larger $N$ will bring the less  codewords for pilot coding and therefore the larger lower bound.
\end{fact}

\begin{remark}
Three points need to be identified. 1) Code partition does not affect the  randomness of codewords; 2)  Code partition does not affect the  randomness of exposed codewords; 3) Code partition  reduces solely the occurrence probability that  both $\bf{c}$ and the exposed codewords of LUs occur in one same segmented submatrix.
\end{remark}
Note that each node does not need to inform Alice of any valuable information, such as, the index of segmented submatrix as Alice is enabled  to identify the index employed because of the multiuser guarantee.

 \begin{theorem}
Given  $N$ segmented submatrices  for each submatrix,  the IEP under PB-PJ attack for  $K$  LUs  using H2DF-$\left( {K,1,B} \right)$ code is  updated and bounded by:
\begin{equation}
\frac{N}{C} \le P \le \frac{1}{{2KN}}
\end{equation}
\end{theorem}
\begin{proof}
 Due to the uncertainty of which segmented submatrix is adopted by the $i_{0}$-th LU,  Eva randomly chooses a codeword in one of the  $N$  segmented submatrices of $\left[ {{{\bf{b}}_{j \in {{\cal B}_{i_{0}}}}}} \right]$ as $\bf{c}$.   In comparison to the worst case without code partition in Theorem 1,  now   the identification error happens iff  there exists an exposed codeword occurring  exactly within  the segmented submatrix employed by  Eva, rather than the original submatrix. We can calculate  IEP   as ${1 \mathord{\left/
 {\vphantom {1 {2K}}} \right.
 \kern-\nulldelimiterspace} {2K}} \times {1 \mathord{\left/
 {\vphantom {1 N}} \right.
 \kern-\nulldelimiterspace} N}$.  For the best case,  the occurrence of duplicate codewords  among   decomposed codewords is  recalculated as  ${{{N}} \mathord{\left/
 {\vphantom {{{N}} C}} \right.
 \kern-\nulldelimiterspace} C}$ since the size of  codebook for pilot coding is reduced to ${C \mathord{\left/
 {\vphantom {C N}} \right.
 \kern-\nulldelimiterspace} N}$.
\end{proof}

 \begin{theorem}
Given $N$ segmented submatrices for each submatrix,  the optimal tradeoff between the upper and lower bounds of IEP for  $K$  LUs under PB-PJ attack is achieved iff:
\begin{equation}
N{\rm{ = }}\sqrt {\frac{C}{2K}}
\end{equation}
The optimal and exact expression  of IEP is derived  by:
\begin{equation}
P=\sqrt {\frac{{\rm{1}}}{{2CK}}}
\end{equation}
In this case, $S_{R}$  and  $R_{S}$ are respectively equal to zero and  $\frac{1}{2}{\log _{10}}2CK$, achieving the  stable and highly-reliable SMPA  performance.
\end{theorem}
By configuring $\frac{N}{C}= \frac{1}{{2KN}}$, we can derive the above optimal $N$.  The codeword CI can also be seen in Algorithm 4.
\subsection{Code Construction and Performance Analysis}
\label{CCPA}
In order to analyze the coding performance in practical communications systems, it is required that  the expression of $C$ should be specified, which depends  on the specific  construction method of cover-free code. We consider the maximum distance separable (MDS) code~\cite{Singleton} based code construction method. The reason is that such a code may
be conveniently  augmented with additional words, without decreasing its distance, hence its order, (or namely the  number $K$ of LUs) by letting the number of ones increases suitably. It means that this method is resilient even when   signals on subcarriers are interfered and thus the number of signals  in detection is changed.  The specific  construction method  can be found in~\cite{Kautz}.  The overall  performance of H2DF-$\left( {K,1,B} \right)$  code  here refers to the code rate, antenna and frequency-domain resource overheads and IEP.

The rate of H2DF-$\left( {K,1,B} \right)$ code  of  length  $B$ and cardinality $C$, denoted by $R\left( {C,B} \right)$, is defined in ~\cite{yachkov} by:
\begin{equation}
R\left( {C, B} \right){\rm{ = }} {{{{{\log }_2}C}} \mathord{\left/
 {\vphantom {{\left( {{{\log }_2}C} \right)} {B}}} \right.
 \kern-\nulldelimiterspace} {B}}
 \end{equation}
 Under  MDS code based code construction, the size of  $\bf B$  satisfies $ B = N^{\rm L}_{\rm{P}}= q\left[ {1 +  K\left( {k - 1} \right)} \right],C = {q^k}, q \ge  K\left( {k - 1} \right)  \ge 3, K\ge 2$, which mathematically denotes  the frequency-domain resource overheads.  Considering  the H2DF encoding   process,  we notice that  the detection process   depends on  $N_{\rm T}$ and  $K$ and the codebook formulation  process depends on $N^{\rm L}_{\rm P}$ and $K$. Therefore, we  introduce the function relationship among $N_{\rm T}$, $N^{\rm L}_{\rm P}$ and  $K$ as follows:
\begin{equation}
\begin{array}{l}
\gamma \left( {\varepsilon ^*=0} \right) = f\left( {{N_{\rm{T}}},K,{\varepsilon ^*=0}} \right), N^{\rm L}_{\rm{P}} = q\left[ {1 +  {K }\left( {k - 1} \right)} \right], \\
 q \ge  {K}\left( {k - 1} \right)\ge 3, K \ge 2.
\end{array}
\end{equation}
where  the function $f$ is the one defined by the  Eq. (49) in~\cite{Kobeissi}.
\begin{figure*}[!t]
{
  \includegraphics[width=2.55in]{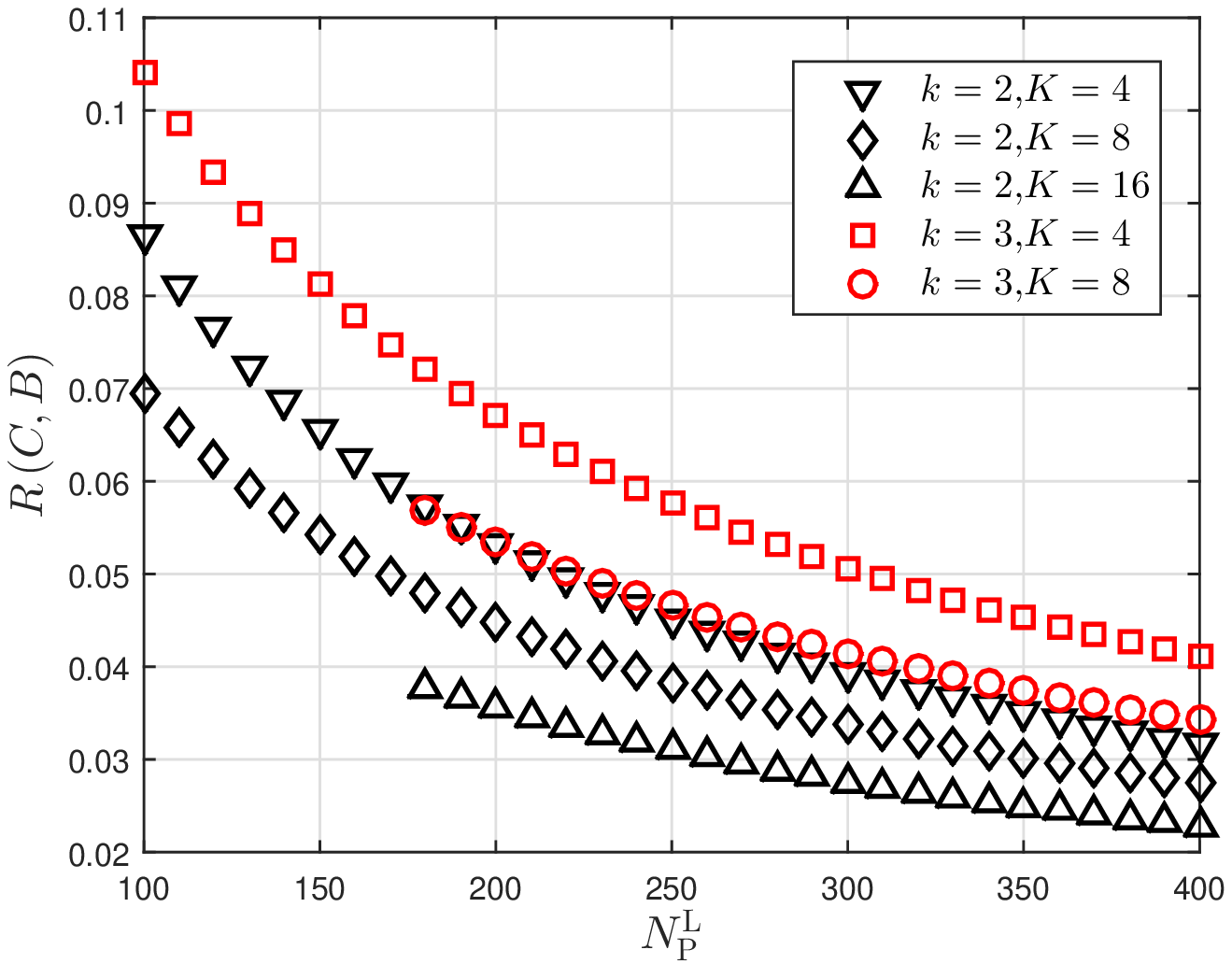} \hspace{-18pt} \includegraphics[width=2.55in]{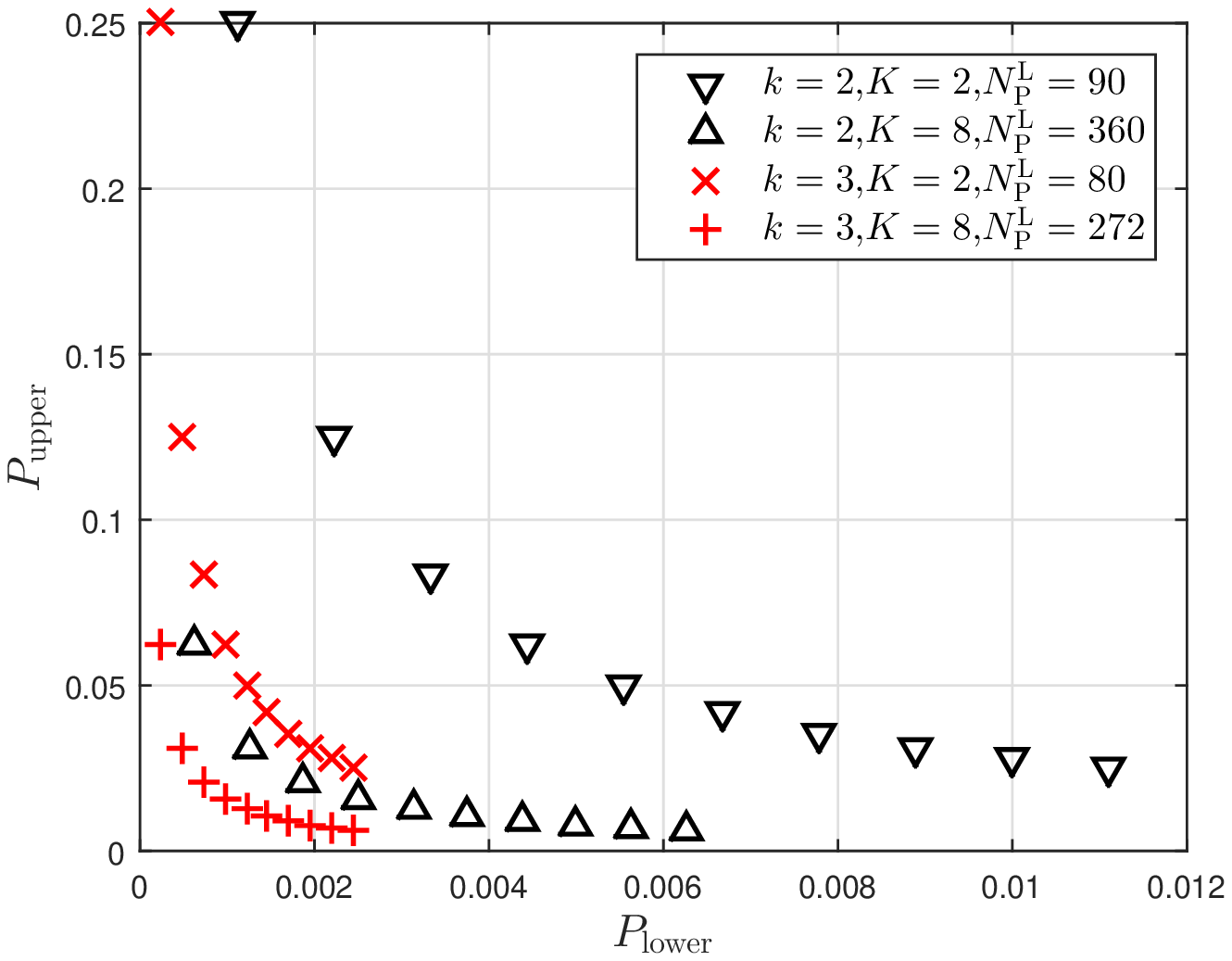}\hspace{-13pt}  \includegraphics[width=2.55in]{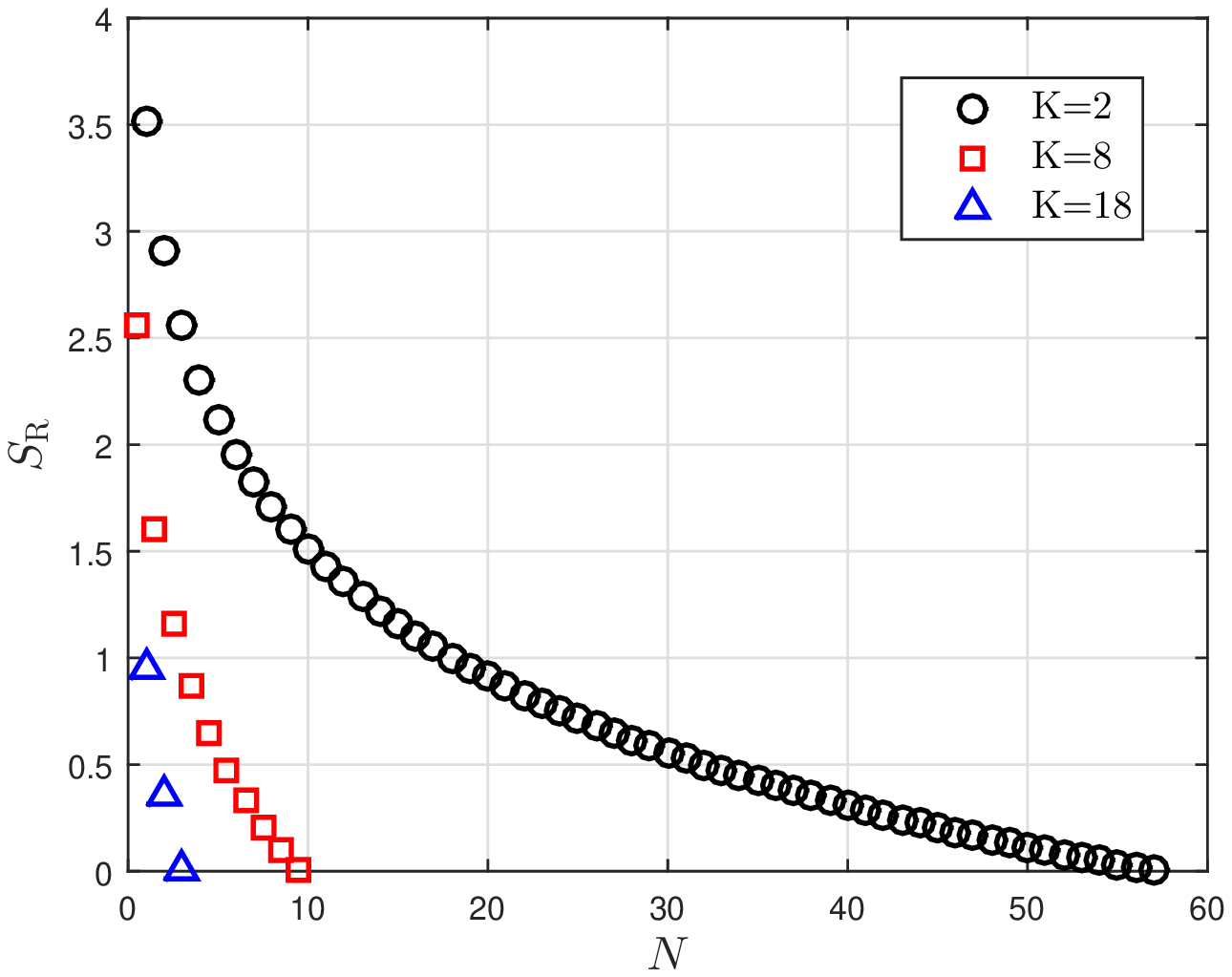}}
  \centerline{(a)\hspace{160pt}(b)\hspace{160pt}(c)}\vspace{-5pt}
  \caption{(a) The code rate $R\left( {C,{\rm{B}}} \right)$ versus $N^{\rm L}_{\rm P}$ under various $k$ and $K$; (b) The upper-lower bound tradeoff  curve under various $K$ and $N^{\rm L}_{\rm P}$; (c) The instability in SMPA reliability  versus $N^{\rm L}_{\rm P}$ under various $k$ and $K$.}
  \label{Simulations}
\end{figure*}
\begin{figure*}[!t]
{
  \includegraphics[width=2.55in]{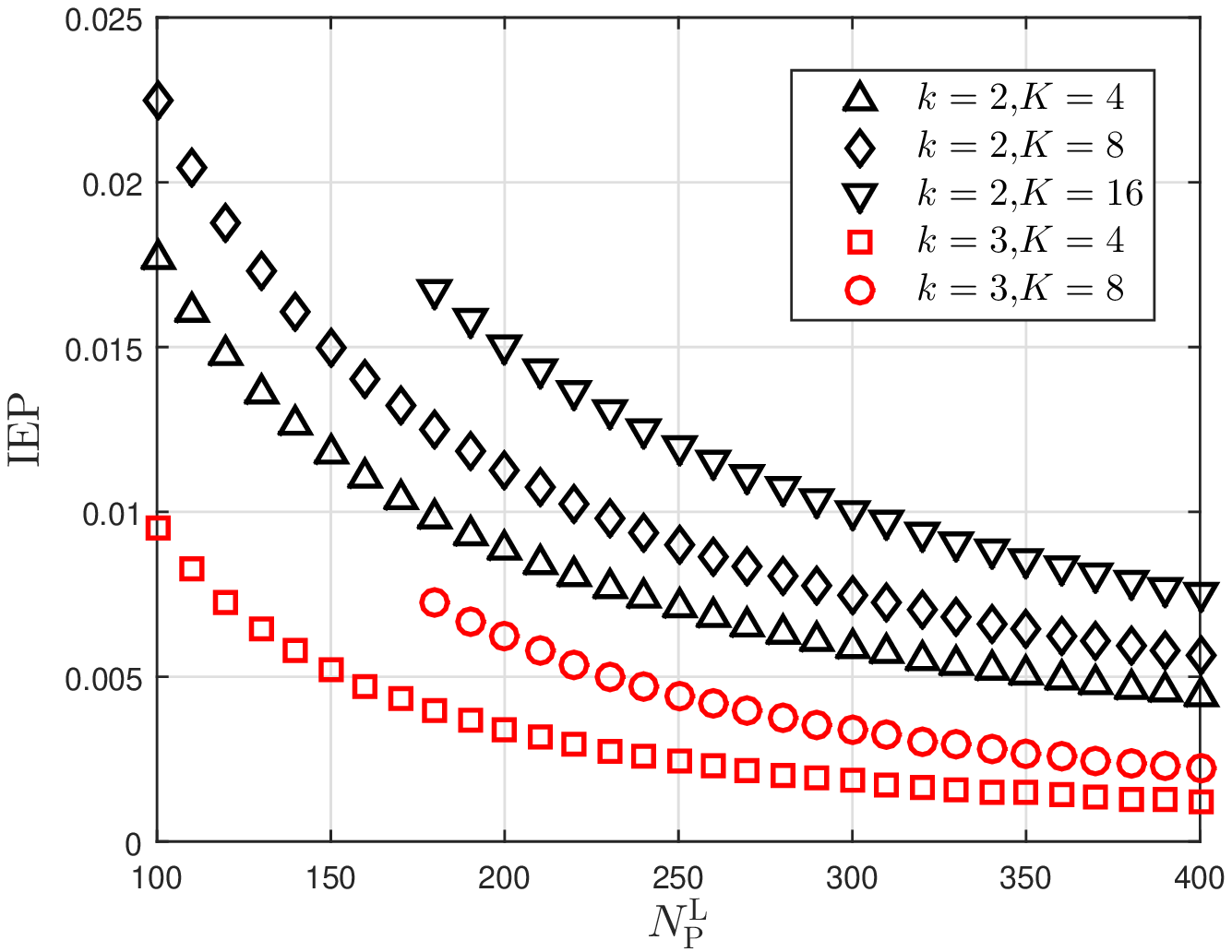} \hspace{-18pt} \includegraphics[width=2.55in]{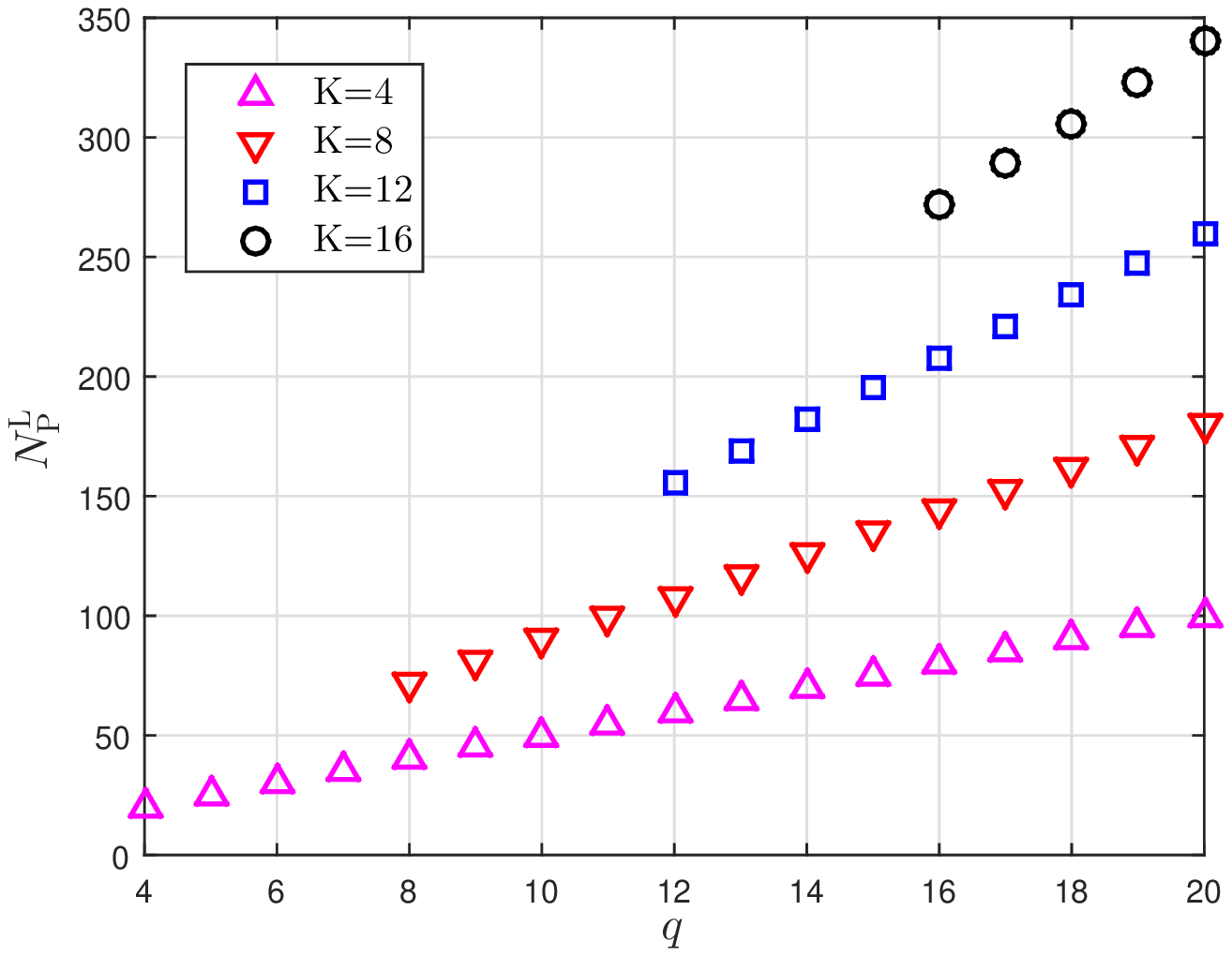}\hspace{-13pt}  \includegraphics[width=2.55in]{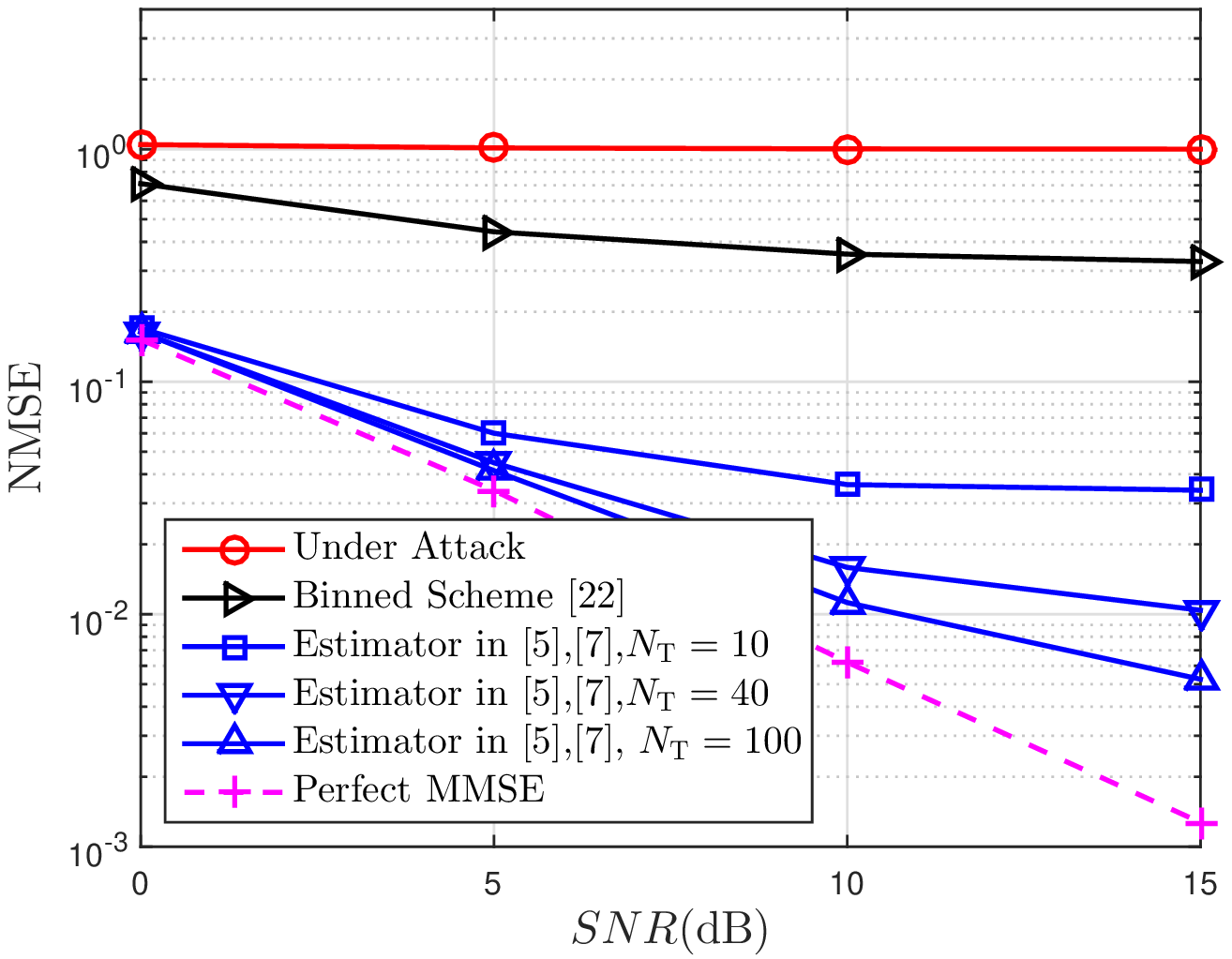}}
  \centerline{(a)\hspace{160pt}(b)\hspace{160pt}(c)}\vspace{-5pt}
  \caption{(a) The  IEP under optimal tradeoff  versus $N^{\rm L}_{\rm P}$ with  various $k$ and $K$; (b) The frequency domain overheads under various $K$; (c) The NMSE versus SNR of LUs  under various  $N_{\rm T}$.}
  \label{Simulations1}
\end{figure*}
 The IEP,  depending on $N^{\rm L}_{\rm P}$ and $K$, thus can  be formulated as follows:
 \begin{theorem}
 With  MDS based code construction  method for H2DF-$\left( {K,1,B} \right)$ code,  IEP can be simplified into:
\begin{equation}\label{E.36}
P=\sqrt {\frac{{{{\left[ {1 + K\left( {k - 1} \right)} \right]}^k}}}{{2{{\left( {N_{\rm{P}}^{\rm{L}}} \right)}^k}K}}}
\end{equation}
where
$N_{\rm{P}}^{\rm{L}} \ge K\left( {k - 1} \right)\left[ {1 + K\left( {k - 1} \right)} \right]$,  $K\left( {k - 1} \right)\ge 3$.
 \end{theorem}
Using $C = {q^k}$,  Eq. (19) and (21), we can easily derive the above theorem.
\section{Simulation Results}
\label{SR}
 \begin{table}[!t]\footnotesize
\caption{\label{tab:test}Simulation Parameters and Values}
\vspace{-15pt}
\begin{center}
\footnotesize
\begin{tabular*}{9cm}{@{\extracolsep{\fill}}ll}
  \toprule
 \multicolumn{1}{c}{ Simulation Parameters}  & \multicolumn{1}{c}{ Values }\\
  \midrule
  \multicolumn{1}{c}{City scenario}&\multicolumn{1}{c}{Urban/non-line-of-sight (NLOS)}
  \\
  \multicolumn{1}{c}{Antennas at BS}&\multicolumn{1}{c}{ Uniform linear array (ULA)}
  \\
  \multicolumn{1}{c}{Maximum number of LUs}&\multicolumn{1}{c}{$K\le 19$}
  \\
  \multicolumn{1}{c}{  Channel fading and scattering model}&\multicolumn{1}{c}{Rayleigh and One-ring~\cite{Xu_Optimal}}
  \\
  \multicolumn{1}{c}{Carrier  frequency }&\multicolumn{1}{c}{$f=$2GHz}
  \\
  \multicolumn{1}{c}{Bandwidth}&\multicolumn{1}{c}{20MHz}
  \\
    \multicolumn{1}{c}{Coherence bandwidth}&\multicolumn{1}{c}{20KHz}
  \\
    \multicolumn{1}{c}{Coherence time }&\multicolumn{1}{c}{${T_{\rm{c}}} = {c \mathord{\left/
 {\vphantom {c {\left( {fv} \right)}}} \right.
 \kern-\nulldelimiterspace} {\left( {fv} \right)}}$}
  \\
  \multicolumn{1}{c}{Maximum velocity of LUs }&\multicolumn{1}{c}{ $v$=360km/h}
  \\
    \multicolumn{1}{c}{Available subcarriers for H2DF coding}&\multicolumn{1}{c}{$N^{\rm A}_{\rm P}=N^{\rm L}_{\rm P}\le {{1200} \mathord{\left/
 {\vphantom {{1200} 3}} \right.
 \kern-\nulldelimiterspace} 3}$; }
  \\
    \multicolumn{1}{c}{Available subcarriers for CIR estimation}&\multicolumn{1}{c}{$N^{\rm A}_{\rm E}=N^{\rm L}_{\rm E}\le 128$;}
  \\
      \multicolumn{1}{c}{Number of channel taps}&\multicolumn{1}{c}{$L_{\rm s}=6$;}
  \\
    \multicolumn{1}{c}{Pilot insertion mode for  CIR estimation }&\multicolumn{1}{c}{Block type}
  \\
  \multicolumn{1}{c}{Channel estimator}&\multicolumn{1}{c}{\cite{Xu_Optimal}}
  \\
  \multicolumn{1}{c}{Modulation}&\multicolumn{1}{c}{OFDM with normal CP}
  \\
    \multicolumn{1}{c}{ Slot structure}&\multicolumn{1}{c}{1 slot= 7 OFDM symbols }
  \\
  \bottomrule
 \end{tabular*}
\end{center}
\vspace{-15pt}
\end{table}
 In this section, we simulate the performance of SMPA  from  two main perspectives, that is, the  coding  perspective and the CIR estimation perspective.  For the former, we focus on five metrics, that is, the code rate,  upper-lower bound tradeoff  curve,  the curve of instability variations, the  IEP curve under the optimal  bound tradeoff and the overheads in coding. We are more concerned about the configuration of system parameters, i.e., $k$, $K$, $q$, $N^{\rm L}_{\rm P}$,  on the impact of those metrics. Usually, $k$ is set to be 2 and 3, which is enough under the practical system configuration. We should also note that $N^{\rm L}_{\rm P}$ is a function of $K$ and $q$. Based on the constraint of the prime power $q$, we know that  $N^{\rm{L}}_{\rm{P}}$ is bounded by $ K\left( {k - 1} \right)\left[ {1 + K\left( {k - 1} \right)} \right]$. Under this condition, we configure $K$ and $k$  artificially and examine the influence of variations of $N_{\rm{P}}^{\rm{L}}$ on the code rate. To evaluate other four metrics, we  assume that  $N^{\rm L}_{\rm P}$ always achieves its lower bound, that is, $K\left( {k - 1} \right)\left[ {1 + K\left( {k - 1} \right)} \right]$.   For the channel estimation part,  we consider the basic configuration shown in Table II.

Fig.~\ref{Simulations} (a) presents the curve of code rate versus $N^{\rm L}_{\rm P}$. It indicates us three facts: 1) Increasing $K$ and $N^{\rm L}_{\rm P}$ will lower down the code rate;  For example, the code rate ranges from 0.04 to 0.03 when increasing $K$ from 4 to 16 at  $k=2$ and $N^{\rm L}_{\rm P}=300$. 2) Increasing  $k$ will  increase the code rate; For example, the code rate increases from 0.03 to 0.04 if $k$ is increased from 2 to 3 when  $K=4$ and $N^{\rm L}_{\rm P}=400$.3)  Increasing  $K$  will increase the  overheads of   $N^{\rm L}_{\rm P}$ since the lower bound of  available $N^{\rm L}_{\rm P}$, that is, $K\left( {k - 1} \right)\left[ {1 + K\left( {k - 1} \right)} \right]$,   increases with the increase of $K$. For example, when $k=3$, the increase of $K$ from 4 to 8 will bring the lower bound of $N_{\rm{P}}^{\rm{L}}$  increasing  from 72 to 272.

 Fig.~\ref{Simulations}(b) presents the  upper-lower bound tradeoff  curve.  We plot ten discrete points  on which there exists a relationship between the upper bound $\frac{1}{{2KN}}$ and lower bound $\frac{N}{C}$. $N$  is configured  from 1 to 10. The reason is that $N$ should be no more than  the optimal value, i.e., $N{\rm{ = }}\sqrt {\frac{C}{2K}}$.  In this context, we configure  $k$ to be 2 and 3 and  $K$ to be  2 and 8.   As we can see, there exists a tradeoff curve on which the upper bound  has to be sacrificed to maintain a certain level of lower bound.

Fig.~\ref{Simulations} (c) presents the curves of instability in SMPA reliability  versus the number of segmented submatrices  under $K=2, 8, 18$. Considering Eq.(16) and Eq.(17), we can know that $S_{\rm R}$ is equal to:
\begin{equation}
{S_{\rm{R}}} =  - 2{\log _{10}}N + G
\end{equation}
where $G = k{\log _{10}}B - k{\log _{10}}\left[ {1 + K\left( {k - 1} \right)} \right] - {\log _{10}}2K$ and $1 \le N \le \frac{{N_{\rm{P}}^{\rm{L}}}}{{K + 1}}\left[ {\sqrt {\frac{1}{{2K}}} } \right]$. From the curves, we can see that the increase of $K$ makes the instability approach  zero more quickly.  This demonstrates that our proposed RBC theory is very suitable for multiuser environment.

Fig.~\ref{Simulations1}(a) presents the value of IEP versus  $N^{\rm L}_{\rm P}$.   It indicates us three facts: 1) Increasing $k$ and $N^{\rm L}_{\rm P}$ will lower down the IEP;  For example, IEP  ranges from $7.5 \times {10^{ - 3}}$ to $3.4 \times {10^{ - 3}}$  with the increase of $k$ from 2 to 3 when $K=8$ and $N^{\rm L}_{\rm P}=300$. When $k=3$ and $K=8$, IEP decreases from $4.4 \times {10^{ - 3}}$ to $2.2 \times {10^{ - 3}}$ with the increase of $N^{\rm L}_{\rm P}$ from 250 to 400. 2) Increasing  $K$ will  increase the IEP; For example, when $N^{\rm L}_{\rm P}=400$ and $k=2$, the IEP  increases from $4.4 \times {10^{ - 3}}$ to $5.6 \times {10^{ - 3}}$ and further to $7.5 \times {10^{ - 3}}$, with the increase of $K$ from 4 to 8 to 16. 3)  Increasing  $K$  will also increase the  overheads of   $N^{\rm L}_{\rm P}$ since the lower bound of  available $N^{\rm L}_{\rm P}$  increases with the increase of $K$.  In literature~\cite{Xu_CF},   the IEP performance is 0.5 if the two conditions hold true: 1) Eva launches randomly-imitating attack after  acquiring $\bf B$ and 2) Its array spatial fading correlation  is not known by Alice, or otherwise Eva has the same mean AoA with LU of interest. In  comparison to the scheme in~\cite{Xu_CF}, our scheme is more robust under the three conditions and able to break  down this IEP floor, i.e., 0.5,  because it is a pure information  coding technique, not depending on the spatial fading correlation models.

In Fig.~\ref{Simulations1}(b), we simulate the coding overheads  in the respect of  $N^{\rm L}_{\rm P}$. Note that the antenna resource overheads  in terms of $N_{\rm T}$ can be seen Fig.~\ref{PF}. We do not simulate it again here. As we can see, $N^{\rm L}_{\rm P}$ increases linearly  with the increase of  the size  of  codebook, i.e., $C$, or  equivalently $q$.  Theoretically,  $N^{\rm L}_{\rm P}$ is a linear function of $q$ when   $q \ge K$.   For example, when $q$ changes from 12  to 20 at $K=8$, namely, $C$ increases from 144 to 400, $N^{\rm L}_{\rm P}$ increases about from 108 to 180.

In Fig.~\ref{Simulations1}(c), we stimulate the performance  of  CIR estimation. Normalized mean squared error (NMSE)  is simulated versus   signal-to-noise ratio (SNR) of LUs under arbitrary SNR of Ava.  For the sake of simplicity,  we assume $\rho _{{\rm{L}}, m}=\rho _{{\rm{L}}},\forall m$.  The performance under  this type of estimator  is not influenced by the specific value of $\rho_{\rm A}$ due to the subspace projection property. We do not consider the case where there is no attack since in this case LS estimator is a natural choice. The CIR estimation  under PTS attack is only presented  since  the estimation error floor  under PTN and PTJ attack can be easily understood to be very high. The binned scheme proposed in~\cite{Shahriar2} is simulated as an another  comparison scheme.   As we can see, attack could cause a high-NMSE floor on  CIR estimation. This phenomenon can also be seen in the binned scheme~\cite{Shahriar2}. However, the estimation in our proposed framework  breaks down this floor and its NMSE  gradually decreases  with the increase of transmitting antennas and gradually  approaches the NMSE curve under perfect  minimum mean-square error (MMSE)  case which serves as   a performance benchmark.
\section{Conclusions}
\label{Conclusions}
In this paper, we designed a H2DF coding theory for  a multi-user multi-antenna OFDM system to protect the pilot authentication process over frequency-selective fading channels. In this scheme, a framework of H2DF coding based  encoding/decoding mechanism  for  randomized pilots  through a HD model was developed, bringing about  the benefits of stable and highly-reliable SMPA. Low IEP was formulated by  upper  and lower bounds.  We verified the tradeoff relationship between   the upper bound and lower bound. We developed a bound contraction theory through which  an optimal upper-lower bound tradeoff can be achieved  using  a codebook partition technique such that the exact bound of IEP can be identified.  Some necessary  performance and simulations results were presented  to verify the robustness  of proposed scheme against pilot aware  attack.

\end{document}